\newif\ifarxiv\arxivtrue
\newif\ifec\ecfalse

\ifec
\documentclass[format=acmsmall, review=false]{acmart}
\usepackage{acm-ec-26}
\fi

\ifarxiv
\documentclass{article}
\usepackage{amsthm,amsmath,amsfonts}
\usepackage{hyperref}
\usepackage[numbers,sort&compress]{natbib}
\usepackage[margin=1in]{geometry}

\theoremstyle{plain}
\newtheorem{theorem}{Theorem}[section]
\newtheorem{proposition}[theorem]{Proposition}
\newtheorem{lemma}[theorem]{Lemma}
\newtheorem{corollary}[theorem]{Corollary}
\fi

\theoremstyle{plain}

\usepackage[ruled,noend]{algorithm2e} 
\usepackage{booktabs} 

\usepackage{cleveref}
\usepackage{thm-restate}
\usepackage{dsfont}
\usepackage{mathtools}
\usepackage{doi}

\usepackage{multirow}
\usepackage{xspace}
\usepackage{tikz}
\usepackage[inline]{enumitem}
\usepackage{graphicx}

\definecolor{myorange}{RGB}{230,159,0}
\definecolor{mylightblue}{RGB}{86,180,233}
\definecolor{mygreen}{RGB}{0,158,115}
\definecolor{myblue}{RGB}{0,114,178}
\definecolor{myyellow}{RGB}{240,228,66}
\definecolor{myred}{RGB}{213,94,0}

\newcommand{\R}{\mathbb{R}}
\newcommand{\x}{\mathbf{x}}
\newcommand{\y}{\mathbf{y}}
\newcommand{\B}{\mathbf{B}}

\renewcommand{\c}{\mathbf{c}}

\newcommand{\N}{\mathbb{N}}
\newcommand{\calR}{\mathcal{R}}

\renewcommand{\d}{\mathrm{d}}

\renewcommand{\P}[1]{\mathbb{P}[ #1 ]}

\newcommand{\CP}[2]{\mathbb{P}[ #1 \;\vert\; #2 ]}

\newcommand{\BiggCP}[2]{\mathbb{P}\Biggl[ #1 \;\Bigg\vert\; #2 \Biggr]}
\newcommand{\bigCP}[2]{\mathbb{P}\bigl[ #1 \;\big\vert\; #2 \bigr]}

\newcommand{\E}[1]{\mathbb{E}[ #1 ]}

\newcommand{\BiggE}[1]{\mathbb{E}\Biggl[ #1 \Biggr]}
\newcommand{\bigE}[1]{\mathbb{E}\bigl[ #1 \bigr]}

\newcommand{\CE}[2]{\mathbb{E}[ #1 \;\vert\; #2 ]}
\newcommand{\BigCE}[2]{\mathbb{E}\Bigl[ #1 \;\Big\vert\; #2 \Bigr]}

\newcommand{\bigCE}[2]{\mathbb{E}\bigl[ #1 \;\big\vert\; #2 \bigr]}

\newcommand{\ninfeasibleSGAP}{\normalfont\textsc{InfeasibleSubGAP}}
\newcommand{\nfeasibleSGAP}{\normalfont\textsc{FeasibleSubGAP}}

\crefname{lem}{Lemma}{Lemmas}
\crefname{prop}{Proposition}{Propositions}

\ifec
\setcitestyle{authoryear}
\fi
\title{Submodular Welfare Maximization with Budget Constraints \newline in the Random-Order Model}

\ifec
\author{Submission 105}
\fi

\ifarxiv
\begin{document}
\author{Max Klimm \and Martin Knaack}
\maketitle
\fi

\begin{abstract}
We study an online item-allocation problem with budgets and a submodular objective. A set of~$m$ agents is known in advance, and each agent~$j$ has a known budget.
A set of $n$~items arrives over time in a uniformly random order. When item~$i$ arrives, its cost~$c_{i,j}$ for each agent~$j$ is revealed, and the algorithm must irrevocably assign~$i$ to an agent without violating any budget constraint.
The goal is to maximize a monotone submodular function defined over all possible assignments $[n]\times[m]$.
At the time of decision, the algorithm has only oracle access to this submodular function restricted to items seen so far.
This model subsumes welfare maximization with submodular valuations, agent-specific item costs, and agent-specific budgets.

We measure the performance of an algorithm by its competitive ratio, i.e., the worst-case ratio between the algorithm’s expected value and that of the offline optimum, which knows all item costs and the full submodular function in advance.
Prior work only considered the case of a single agent and achieved a $\smash{\frac{1}{54.4}}$-competitive algorithm.
We generalize and improve this result to a polynomial-time $\alpha$-competitive algorithm with
$\smash{\alpha=\bigl(1-\frac{1}{e}\bigr)\bigl(\frac{1}{\sqrt e}-\frac{1}{2}\bigr)\approx \frac{1}{14.85}}$ for an arbitrary number of agents.

We also study the special case in which all item costs and all budgets equal~$1$ which yields an online submodular matching problem.
Prior work achieved a polynomial $\frac{1}{9.66}$-competitive  algorithm for this problem; we improve this to $\smash{e^{\sqrt{2}-2}\bigl(\sqrt{2}-1\bigr)\bigl(1-\tfrac{1}{e}\bigr)} \approx \frac{1}{6.86}$.

Both our algorithms rely on repeatedly computing $\bigl(1 - \frac{1}{e}\bigr)$-approximations for the multilinear extension of offline variants of subproblems. If super-polynomial runtime is allowed, these subproblems can be solved optimally, and our competitive ratios improve by this factor.
\end{abstract}

\ifec
\begin{document}
\begin{titlepage}
\maketitle
\end{titlepage}
\fi

\section{Introduction}

Online allocation problems are ubiquitous in market design and algorithmic decision making: Items, tasks, impressions, applicants, or requests arrive over time, and a decision maker must allocate them to agents under hard feasibility constraints. In many canonical settings, an algorithm must commit immediately to an assignment, with limited information about the future.
This online nature creates an inherent trade-off between exploitation (accepting good current offers) and exploration (leaving more capacity for accommodating possible better future offers).
The \emph{competitive ratio} quantifies the limits that online algorithms face due to this trade-off; it is defined as the worst-case ratio between the expected value achieved by the online algorithm and the value achieved by an omniscient offline algorithm that has all relevant knowledge in advance and can thus compute an optimal solution. 

A definitive answer to this trade-off has been first given for the so-called \emph{secretary problem} where the decision maker observes the value of $n$~items in a random order.
After revelation of the value of an item, the decision maker has to decide immediately whether the item is selected or discarded. At most one item can be selected in this process.
For the objective of maximizing the probability of selecting the item with the highest value, the best possible competitive ratio of $1/e$ is attained by the algorithm that discards the first $n/e$ items, and selects the first item that yields a higher value than any other item seen so far \citep{Lindley61,Dynkin63}. Also for the objective of maximizing the expected value of the selected item it achieves the same competitive ratio, and no better competitive ratio is possible \citep{CorreaDFS22}. 

While obviously important, the secretary problem is a very stylized question that falls short in modeling many issues faced in real-world allocation problems.
First, there is often more than a single agent to which allocations can be made.
This is the case, for example, on many online platforms where offers arrive over time and must be allocated to potential customers.
Second, agents may be subject to more complicated feasibility constraints. 
For example, in the AdWords problem \citep{MehtaSVV07} sponsored-search impressions arrive over time and must be allocated to advertiser subject to budget constraints.
Third, when multiple items are allocated to an agent, there often are substitute effects so that the overall value of the agent is a submodular function of the items allocated to them.

In this paper, we study a general model with $m$ agents that are known in advance that we term \emph{submodular generalized assignment problem}.
Each agent~$j$ has a budget~$B_j$ that is also known in advance.
A set of $n$ items arrives in an order picked uniformly at random from all $n!$ permutations of the set of items.
Upon arrival of an item~$i$, its cost $c_{i,j}$ for each agent~$j$ is revealed.
An online algorithm must assign each item~$i$ immediately and irrevocably to one of the agents, or has to decide that the item is not assigned at all.
Items may only be assigned to agents if the total cost of the items assigned to an agent does not exceed their budget.
The goal is to maximize a non-negative, non-decreasing and submodular function $f \colon 2^{[n] \times [m]} \to \mathbb{R}_{\geq 0}$ defined on the space $2^{[n] \times [m]}$ of all assignments of items to agents.
The function~$f$ is initially unknown and when making the decision for item~$i$, the algorithm can only query the values of all assignments involving the items that arrived so far.

\subsection{Our results}

As our main result, we obtain a polynomial-time online algorithm for the submodular generalized assignment problem that is $\alpha$-competitive where
$\smash{\alpha = \bigl(1 - \frac{1}{e}\bigr)\bigl(\frac{1}{\sqrt{e}} - \frac{1}{2}\bigr) \approx \frac{1}{14.85}}$. 
This is the first constant competitive algorithm for the submodular generalized assignment problem.
Our result further improves upon the previously best known competitive ratio of $\smash{\frac{1}{20e} \approx \frac{1}{54.4}}$ for the special case of the \emph{submodular knapsack secretary problem} which corresponds to a submodular generalized assignment problem with a single agent \citep{FeldmanNS11}.

Our algorithm first discards the first $\lfloor n/\sqrt{e} \rfloor$~items.
For each later item, it computes a $\bigl(1 - \frac{1}{e}\bigr)$-approximate solution for the multilinear extension of the problem restricted to all items seen so far by the continuous greedy algorithm of \citet{CalinescuCPV11}.
It then interprets the fractional values obtained in the approximate solution as assignment probabilities.
The algorithm first computes an infeasible assignment with the property that the removal of the last item assigned to the agent turns the assignment feasible.
Before the arrival of any item, the algorithm flips a coin.
Before any items arrive, the algorithm flips a single fair coin.
In one outcome, it retains, for every agent, all assignments preceding the first budget violation.
In the other outcome, it retains only the item
causing that violation for each agent.
This route has been taken before by \citet{KlimmK25} for the generalized assignment problem with \emph{linear} objective. 
The analysis in the submodular case, however, is much more involved.
Intuitively, the contribution of an item to the objective depends on the items that arrived \emph{before} it, but the analysis wants to avoid such dependence.
Instead, we refine an idea used by \citet{KesselheimT17} for the analysis of their algorithms for the submodular secretary and submodular secretary matching problems: We estimate the contribution of each item by its contribution after the arrival of all items that arrive \emph{later}. This avoids double counting of contributions and circumvents dependencies on the arrival order of previous items.

If super-polynomial runtime is allowed, there is no need to consider the multilinear extension, and we can compute an optimal solution of the offline problem restricted to all items seen so far in each step.
This yields a super-polynomial algorithm with a competitive ratio of $\smash{\frac{1}{\sqrt{e}} - \frac{1}{2} \approx \frac{1}{9.39}}$.

Previous work sometimes assumes a large-markets assumptions \citep[e.g.][]{MehtaSVV07,Vaze17} essentially stating that the ratio of each single item to the optimum solution is in $o_n(1)$.
Under this assumption, we can simply run with probability $1$ the part of our mechanism that assigns to each agent all items potentially except the last one that violates the budget constraint. This yields a polynomial algorithm with a competitive ratio of $\smash{\bigl(1-\frac{1}{e}\bigr)\bigl(\frac{2}{\sqrt{e}} - 1\bigr) \approx \frac{1}{7.42}}$ and a super-polynomial algorithm with a competitive ratio of $\smash{\frac{2}{\sqrt{e}} - 1 \approx \frac{1}{4.69}}$.

We further study the special case where the cost of each item and the budget of each agent is $1$, which corresponds to the \emph{submodular secretary matching problem} previously studied by \citet{KesselheimT17}.
They obtained a polynomial-time algorithm that is $\smash{\frac{1}{9.66}}$-competitive (and a super-polynomial algorithm that is $\smash{\frac{1}{4.84}}$-competitive).
We improve this to a polynomial-time algorithm that is $\smash{\frac{1}{6.86}}$-competitive (and a super-polynomial algorithm that is $\smash{\frac{1}{4.34}}$-competitive).
While the main ideas of both algorithms are similar, \citeauthor{KesselheimT17} solve the \emph{integer} offline problem in each step and tentatively assign whenever the current item appears in the solution.
By contrast, we optimize the \emph{multilinear extension} in each step and interpret the fractional values as tentative assignment probabilities.
Our improved competitive ratio for polynomial algorithms stems from to factors: First, the multilinear extension can be approximated by a factor of $\bigl(1-\frac{1}{e}\bigr)$ while the integer optimum can only be approximated by a factor of $\smash{\frac{1}{2+\epsilon}}$ for any $\epsilon > 0$.
Second, our analysis is more careful as witnessed by the slightly better competitive ratio for super-polynomial algorithms.

\subsection{Related work}

To the best of our knowledge, the submodular generalized assignment problem introduced in this paper has not been studied before in the literature. However, it contains several other problems as special cases that we discuss in this section. For a better overview of some of the different problem classes and their relationship, see also Fig.~\ref{fig:related-problems}.

\begin{figure}
\small
\begin{center}
\ifec
\begin{tikzpicture}[yscale=1.22,xscale=1]
\fi
\ifarxiv
\begin{tikzpicture}[yscale=1.5,xscale=1.33]
\fi
\draw[fill=lightgray, fill opacity=0.5,rounded corners=10pt] (0.25,0) rectangle (12,4);
\node[align=center,text width=7cm,anchor=north] (sgap) at (6,3.8) {{\bf submodular generalized assignment}\\ \textcolor{myred}{\bf $1/14.85$}};
%
%
\draw[fill=mylightblue,fill opacity=0.5,rounded corners=10pt] (0.5,0.25) rectangle (7.75,2);
\node[align=center,text width=3cm,anchor=north] (sks) at (2.5,1.625) {{\bf submodular knapsack secretary} \\ $1/54.4 \rightarrow \textcolor{myred}{1/14.85}$};
\draw[fill=mylightblue,fill opacity=0.5,rounded corners=10pt] (4.375,0.375) rectangle (11.625,3);
\node[align=center,text width=3cm,anchor=north] (ssm) at (8,3) {{\bf submodular secretary matching} \\ $1/9.66 \rightarrow \textcolor{myred}{1/6.86}$};
\draw[fill=mygreen,fill opacity=0.5,rounded corners=10pt] (4.25,-1) rectangle (11.75,1.875);
\node[align=center,text width=3cm,anchor=north] (matroid) at (8,-0.2) {{\bf submodular matroid} \\ $1/\log \log \rho$};
\draw[fill=myblue,fill opacity=0.3,rounded corners=10pt] (4.5,0.5) rectangle (11.5,1.75);
\node[align=center,text width=3cm,anchor=north] (adwords) at (9.5,1.625) {{\bf submodular transversal matroid} \\ $1/9.66 \rightarrow \textcolor{myred}{1/6.86}$};
\draw[fill=myyellow,fill opacity=0.5,rounded corners=10pt] (4.625,0.625) rectangle (7.5,1.625);
\node[align=center,text width=3cm,anchor=north] (secretary) at (6,1.625) {{\bf submodular secretary} \\ $1/5.10$};
\end{tikzpicture}
\end{center}
\caption{Problem classes for online problems in the random-order model with submodular objectives and the best known competitive ratios for polynomial algorithms. Numbers in red are achieved in this paper, numbers in black appeared in previous work.
The stated competitive ratios are for monotone submodular functions except the one for the submodular matroid problem.
\label{fig:related-problems}
}
\end{figure}
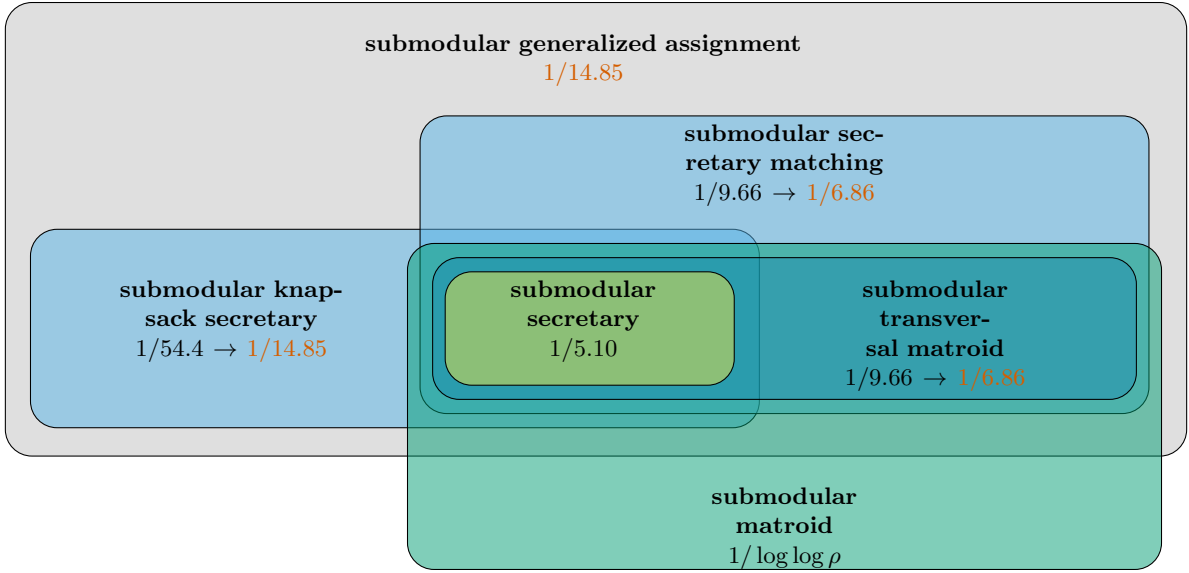

The online \emph{generalized assignment problem} in the random-order model is special case of the problem considered in this paper that appears when the function $f$ is modular.
More specifically, for each item~$i$, there is a value $v_{i,j}$ that is added to the objective when assigning the item to agent~$j$.
This problem has been first studied by \citet{KesselheimRTV18} who gave an algorithm that is $\smash{\frac{1}{8.06}}$-competitive. For the same problem, an algorithm with an improved competitive ratio of~$\smash{\frac{1}{6.99}}$ has been given by \citet{NaoriR19} and, independently, by \citet{AlbersKL21}. An algorithm with a further improved competitiveness of~$\smash{\frac{1}{6.52}}$ has been devised by \citet{KlimmK25}.
For the offline variant of the generalized assignment problem (with linear objective), \citet{FleischerGMS11} gave an algorithm with an approximation guarantee of $\smash{1-\frac{1}{e}}$. \citet{FeigeV06} improved this approximation factor slightly to $\smash{1 - \frac{1}{e} + \epsilon}$ for a small absolute constant $\epsilon > 0$.
There is no $\frac{10}{11}$-approximation for the generalized assignment problem, unless $\mathsf{P}=\mathsf{NP}$, as shown by \citet{ChakrabartyG10}.

A special case of the generalized assignment problem is the AdWords problem first studied by \citet{MehtaSVV07}. 
Here, the values and the costs coincide, i.e., $c_{i,j} = v_{i,j}$ for all $i,j$. 
They obtained a competitive ratio of $\smash{1 - \frac{1}{e}}$ under a large markets assumption that the cost of each item is small compared to the budget, but assuming adversarial rather than random order.
\citet{GoelM08} showed that a competitive ratio of~$\smash{1 - \frac{1}{e}}$ is also obtained by the greedy algorithm in the random-order model. They also use a variant of a large markets assumption. 

The offline version of this problem is also known as the \emph{maximum budgeted allocation problem}.
Without the large markets assumption, for this problem a $\bigl(1-\frac{1}{e}\bigr)$-approximation was given by \citet{AndelmanM04}.
This has been improved to a $\frac{2}{3}$-approximation by \citet{AzarBKMN08} and further to a $\frac{3}{4}$-approximation by \citet{Srinivasan08} and, independently, by \citet{ChakrabartyG10}.
\citeauthor{ChakrabartyG10} further showed that it is $\mathsf{NP}$-hard to provide an approximation better than $\frac{15}{16}$.

Another special case of the problem considered in this paper is when each agent~$j$ has a non-negative, non-decreasing and submodular function $f_j \colon 2^{[n]} \to \mathbb{R}_{\geq 0}$ and the value of an assignment $A \subset [n] \times [m]$ is defined as $\sum_{j \in [m]} f_j(\{i \in [n] \colon (i,j) \in A\})$.
This problem is known as the \emph{submodular welfare problem}. For the online variant of this problem (even with adversarial item arrivals), the natural greedy algorithm that assigns the item to the agent for which it achieves the highest marginal increase in value is $\frac{1}{2}$-competitive \citep{FisherNW78,LehmannLN06}. This competitive ratio is best-possible, unless $\mathsf{NP} = \mathsf{RP}$ as shown by \citet{KapralovPV13}.
As discussed by \citeauthor{FisherNW78}, there is a natural partition matroid defined on $[n] \times [m]$ that captures the constraint that each item can be allocated to one agent only, so submodular welfare maximization is a special case of submodular function maximization under a matroid constraint.
For the offline variant of submodular welfare approximation, \citet{Vondrak08} gave a $\bigl(1-\frac{1}{e}\bigr)$-approximation. This approximation is best-possible, unless $\mathsf{P} = \mathsf{NP}$ as shown by \citet{KhotLMM08}.
\citet{CalinescuCPV11} gave a $\smash{\bigl(1-\frac{1}{e}\bigr)}$-approximation for the more general problem of maximizing a submodular function under a matroid constraint.
Their result further implies a $\bigl(1- \frac{1}{e}-o(1)\bigr)$-approximation for the offline generalized assignment problem.

Another special case of the problem considered in this paper is the \emph{submodular knapsack secretary problem} which corresponds to the case of a single agent.
For this problem, \citet{BateniHZ13} gave an constant-competitive algorithm.
\citet{FeldmanNS11} computed an explicit value for the competitive ratio for the algorithm by \citeauthor{BateniHZ13} and gave an improved competitive ratio of $\frac{1}{20e} \approx \frac{1}{54.4}$. 
\citeauthor{BateniHZ13} further claim a $O(1)$-competitive algorithm for the non-monotone case, but the proof is only sketched.
\citet{AmanatidisKS22} gave a full proof for a $\smash{\frac{1}{1710}}$-competitive algorithm for this case.
\citet{AmanatidisMSSTT25} further studied this problem in a learning-augmented setting where a mechanism also obtains a prediction for the optimal offline value. Even when the prediction is arbitrarily bad, they obtained a competitive ratio of~$\smash{\frac{1}{445}}$. 

A special case of the submodular knapsack secretary problem is the
\emph{submodular (cardinality) secretary problem} which corresponds to the case of a single agent and items of unit cost.
The first constant competitive algorithm for this problem was given by \citet{BateniHZ13} who achieved a competitive ratio of $\smash{\frac{1-1/e}{7} \approx \frac{1}{11.1}}$.
\citet{FeldmanNS11} improved the competitive ratio to $\smash{\frac{e-1}{e^2+e} \approx \frac{1}{5.88}}$, and  \citet{KesselheimT17} further improved the competitive ratio to $0.31(1-1/e) \approx \frac{1}{5.10}$.
For the case that a non-monotone submodular function is allowed, \citet{BateniHZ13} obtained a competitive ratio of $\smash{\frac{1}{8e^2} \approx \frac{1}{59.1}}$.

Another special case of the submodular generalized assignment problem studied in this paper is the \emph{submodular secretary matching problem}. It corresponds to the case where the cost of all items and the budget of each agent is equal to $1$.
For this problem, \citet{KesselheimT17} gave a~$0.207\beta$-competitive algorithm where $\beta$ is the approximation guarantee of a polynomial algorithm for the offline variant of the problem.
Observing that bipartite matching can be phrased as the intersection of $2$~matroids and using the $\frac{1}{2+\epsilon}$-approximation by \citet{LeeSV10} for every~$\epsilon > 0$, this yields a polynomial $\smash{\frac{1}{9.66}}$-competitive algorithm for submodular secretary matching.\footnote{An earlier arXiv version \citep{KesselheimT16} claimed a slightly better~$\beta/4$-approximation, but we are not sure whether this result is correct and, thus, report here the bound from the peer-reviewed paper \citep{KesselheimT17}.}
The submodular secretary matching problem is more general than the \emph{submodular transversal matroid secretary problem} for which \citet{MaTW16} obtained a $\smash{\frac{1}{95}}$-competitive algorithm.

Both the submodular transversal matroid secretary problem and the submodular secretary problem are special cases of the \emph{submodular matroid secretary problem}. For the latter, \cite{FeldmanZ18} obtained a $O(1/\log \log \rho)$-competitive algorithm where $\rho$ is the rank of the matroid.
They do not assume that the submodular function is monotone.

Further related is the work of \citet{TangWC22} who studied $k$-submodular secretary problems.
The $k$-submodularity property is defined over $k$-disjoint subsets of a ground set and, loosely speaking, requires that the marginal gain of adding an element to one of the disjoint sets  is smaller when the disjoint sets are larger.
It is not hard to convince ourselves that monotone $k$-submodularity is more general than assuming that each of the $k$~sets has its own submodular and monotone set function and taking the sum of them.\footnote{In fact, by a characterization of \citet[]{WardZ16} a function is $k$-submodular if and only if it is submodular in every orthant and pairwise monotone. Pairwise monotonicity is trivially satisfied as we assume that the $k$-submodular function is monotone. Submodularity in every orthant means that the marginal increase of adding an element $a \notin S_1 \dot\cup S_2 \dot\cup \dots \dot\cup S_k$ to $S_i$ with $i \in [k]$ is not larger than the marginal increase when adding it to $T_i$ when the sets are $T_1 \dot\cup T_2 \dot\cup \dots \dot\cup T_k$ with $T_j \subseteq S_j$ for all $j \in [k]$. The latter condition is clearly fulfilled when having a separate submodular function for each set and maximizing the sum.}
In the problems they consider, upon arrival of an element, it has to be immediately assigned to one of the~$k$ disjoint sets or discarded.
For the problem where a constraint on the cost of each of the $k$~sets is given, they give an algorithm whose competitiveness depends on the sum of the cost constraints. This problem is more general than the submodular generalized assignment problem when all item costs are one, but it fails to capture knapsack-like constraints. \citeauthor{TangWC22} further consider the maximization of a monotone $k$-submodular function under a single knapsack constraint and achieve a polynomial $\frac{1}{313}$-approximation.
However, they only consider a single knapsack constraint over all sets which does not allow us to model individual knapsack constraints as we do.

\section{Preliminaries}

For a natural number $n \in \mathbb{N}$, we denote by $[n] = \{1,\dots,n\}$ the set of the first $n$ natural numbers.
For a finite set $E$, we denote by $\mathbb{R}^E$ the set of $|E|$-dimensional vectors (implicitly assuming a fixed order on $E$), and by $\mathbf{1}_e$ the unit vector in that space that has an entry~$1$ only on the position corresponding to $e \in E$.
Next, we introduce submodular functions, the multilinear extension of a submodular function, and we formally introduce the submodular generalized assignment problem.

\subsection{Submodular Functions and the Multilinear Extension}

Given a finite ground set $ E $, a set function $f \colon 2^E \to \R$ is said to be
\begin{enumerate}[label=(\roman*)]
    \item normalized, if $f(\emptyset) = 0$,
    \item monotone, if $ f(A) \geq f(B) $ for every $B \subseteq A \subseteq E $,
    \item submodular, if $ f(A) + f(B) \geq f(A \cup B) + f(A \cap B)$ for every $ A, B \subseteq E $.
\end{enumerate}
For $ A, B \subseteq E$ we let $f(A \mid B) = f(A \cup B) - f(B)$ denote the marginal increase of $A$ with respect to~$B$.
For a single element $ e \in E$ and $ A \subseteq E$, we write $f(e)$ instead of $f(\{e\})$ and $f(e \mid A)$ instead of $f(\{e\} \mid A)$. 
A set function $f \colon 2^E \to \R$ is monotone if and only if $f(e \mid A) \geq 0$ for every $A \subseteq E$ and $e \in E$ and it is submodular if and only if it satisfies the diminishing returns property, i.e., $f(e \mid A) \leq f(e \mid B)$ for every $ B \subseteq A \subseteq E $ and $e \in E \setminus A$.
For further characterizations, see \citet{NemhauserWF78}.

Any set function $f$ defined on a finite ground set $E$ can also be seen as a function over binary points $ \{0,1\}^E $. 
An extension of a function $f \colon \{0,1\}^E \to \R$ to the hypercube $[0,1]^E$ is a function $F \colon [0,1]^E \to \R$ that satisfies $f(\x) = F(\x)$ for every $\x = (x_e)_{e \in E} \in \{0,1\}^E$.
As defined by \citet[Definition~2.1]{CalinescuCPV11}, a function $F \colon [0,1]^E \to \R$ is smooth monotone submodular if 
\begin{enumerate}[label=(\roman*)]
    \item $ F \in C^2([0,1]^E)$, i.e., $F$ has continuous second-order partial derivatives on $[0,1]^E$,
    \item for every $e \in E$ and $ \x \in [0,1]^E $, we have $ \smash{\frac{\partial F}{\partial x_e}(\x) \geq 0}$,
    \item for every $e,e' \in E$ and $ \x \in [0,1]^E $, we have $ \smash{\frac{\partial^2 F}{\partial x_e \partial x_{e'}}(\x) \leq 0}$.
\end{enumerate}
The second condition implies monotonicity, i.e., for every $\x,\y \in [0,1]^E $ with $\x \geq \y$ componentwise, we have $F(\x) \geq F(\y)$. 
Further, the third condition implies a continuous analogy of submodularity, which states for every $\x,\y \in [0,1]^E$ that $ F(\x) + F(\y) \geq F(\x \lor \y) + F(\x \land \y) $, where the operators $ \lor $ and $ \land $ are defined componentwise by $ (\x \lor \y)_e = \max\{x_e,y_e\} $ and $ (\x \land \y)_e = \min\{x_e,y_e\} $ for every $e \in E$.
We also define a marginal increase for~$F$ by $F(\x \mid \y) = F(\x \lor \y) - F(\y)$ for all $\x, \y \in [0,1]^E$.

An extension that has been extensively used for the maximization of submodular functions is the multilinear extension introduced by \citet{CalinescuCPV11}. For a given function $f \colon 2^E \to \R $, the multilinear extension $F \colon [0,1]^E \to \R$ of $f$ is defined for every $\x = (x_e)_{e \in E} \in [0,1]^E$ by $ F(\x) = \E{f(X)} $ where $X \subseteq E$ is a random set that contains each $e \in E$ independently with probability $x_e$. Then, it holds that
\begin{equation*}
    F(\x) = \sum_{R \subseteq E} \P{X = R} \, f(R) = \sum_{R \subseteq E} f(R) \, \prod_{e \in R} x_e \, \prod_{e \notin R} (1-x_e).
\end{equation*}
The same authors show that the multilinear extension $F$ of a monotone and submodular set function~$f$ is smooth monotone submodular. 
The given properties can be used to generalize well-known inequalities for monotone and submodular set functions.
We state such an inequality in the following proposition. For the proof of this inequality, we refer to \citet[Proposition~1]{Wolsey82}.

\begin{proposition}
    \label{prop:extended-submodular-inequality}
    Let $F \colon [0,1]^E \to \R$ be a smooth monotone submodular function. Then, we have for every $\x,\y \in [0,1]^E$ that
            $F(\x) \leq F(\y) + \sum_{e \in E: x_e > y_e} F( x_e \cdot \mathbf{1}_e \mid \y )$.
\end{proposition}

Based on the multilinear extension $F$ of a given monotone and submodular function $f$, \citet{CalinescuCPV11} developed the continuous greedy algorithm. We summarize one of their results in the following.

\begin{lemma}[\citet{CalinescuCPV11}]
    \label{lem:continuous-greedy}
    For $n \in \N$, let $F \colon [0,1]^n \to \R$ be the multilinear extension of a monotone and submodular function $f$ and let $P \subseteq [0,1]^n$ be a solvable polytope\footnote{A polytope $ P \subseteq [0,1]^n $ is solvable if we can efficiently do linear optimization over $P$, or more formally, for every $ \c \in \R^n $, there is a polynomial-time algorithm that solves $ \arg\max_{\x \in P} \c^\top \x $}. The continuous greedy algorithm computes $\y \in P$ such that $\smash{F(\y) \geq \bigl(1 - \frac{1}{e}\bigr) \max_{\x \in P} F(\x)}$.
\end{lemma}

\subsection{Submodular Generalized Assignment Problem}

In the submodular generalized assignment problem each agent $j \in [m]$ has a budget~$B_j \in \R_{> 0}$ and every item has an agent-dependent cost given by a vector $\c_j = (c_{i,j})_{i \in [n]}$. 
The goal is to select $m$ disjoint sets of items, one for each agent, such that the total cost of a set given to an agent does not exceed its budget and such that a normalized monotone submodular function $f$ is maximized.
We let $\x = (x_{i,j})_{i \in [n],j \in [m]} \in \{0,1\}^{n \times m}$ denote the assignment matrix, i.e., $x_{i,j} = 1$ if item $i \in [n]$ is assigned to agent $j \in [m]$ and $x_{i,j} = 0$ otherwise.
Then, with $f \colon \{0,1\}^{n \times m} \to \R$, we can describe the problem as follows:
\begin{subequations}
\label{eq:main-problem}
\begin{alignat}{3}
    & \text{maximize }& f(\x) \hskip2.5em &&& \nonumber \\
    & \text{subject to }& \sum_{i \in [n]} c_{i,j} \, x_{i,j} &\leq B_j &\quad &\text{for every } j \in [m],  \label{eq:SubmodularGapConstraint1}\\
    &&\sum_{j \in [m]} x_{i,j} &\leq 1 &\quad &\text{for every } i \in [n], \label{eq:SubmodularGapConstraint2}\\
    &&x_{i,j} &\in \{0,1\} &\quad &\text{for every } i \in [n], j \in [m]. \label{eq:SubmodularGapConstraint3} 
\end{alignat}
\end{subequations}
We refer to this problem as the submodular generalized assignment problem.
Without loss of generality, we assume that $c_{i,j}\leq B_j$ for all $i\in[n]$ and $j\in[m]$. Indeed, assignments with $c_{i,j}>B_j$ cannot
occur in any feasible solution.
We may consider a different instance where $c_{i,j} = B_j$ and $x_{i,j}$ does not contribute to the objective, so that we can ignore this assignment when computing optimal solutions.
If we denote the multilinear extension of $f$ by $F \colon [0,1]^{n \times m} \to \R$, we can also consider the following relaxation of~\eqref{eq:main-problem}:
\begin{align}
    \begin{aligned}
        &\text{maximize } F(\x) \quad &&\text{subject to \eqref{eq:SubmodularGapConstraint1}, \eqref{eq:SubmodularGapConstraint2}, and } \\
        &&& \phantom{\text{subject to }} x_{i,j} \geq 0 \quad \text{for every } i \in [n], j \in [m].
    \end{aligned}
\end{align}
We remark that submodular generalized assignment problem is a generalization of the budgeted submodular welfare maximization problem, where each agent~$j \in [m]$ has a normalized monotone submodular valuation function $f_j \colon \{0,1\}^n \to \R $ and the goal is to maximize the social welfare $ \sum_{j \in [m]} f_j(\x_j) $, where $\x_j = (x_{i,j})_{i \in [n]}$, subject to the constraints \eqref{eq:SubmodularGapConstraint1}, \eqref{eq:SubmodularGapConstraint2}, and \eqref{eq:SubmodularGapConstraint3}. The relationship holds since $f(\x) = \sum_{j \in [m]} f_j(\x_j)$ defines a normalized monotone and submodular function $f$.

Specifically, we consider an online variant of the submodular generalized assignment problem where the items are revealed in a random order. 
At the beginning, an algorithm only knows the total number of items $n$ and the budget~$B_j$ of every agent~$j \in [m]$. 
The items are then revealed to the algorithm in an order that is drawn uniformly at random from the set of all permutations $\Pi = \{\pi \mid \pi \colon [n] \to [n] \text{ is a permutation}\}$.
In every round $\ell \in [n]$, an item~$i = \pi(\ell)$ is revealed, and the algorithm observes its cost vector~$\c_i = (c_{i,j})_{j \in [m]}$ and may query the function~$f$ for assignments of the revealed items $\{\pi(1),\dotsc,\pi(\ell)\}$.
The algorithm then has to decide immediately and irrevocably to which agent the item is assigned, or to leave the item unassigned.

We evaluate the performance of such an online algorithm by \emph{competitive analysis}. We say that an algorithm in the random-order model is \emph{$\alpha$-competitive} for some $\alpha \in [0,1]$, if, for every instance, the assignment $\x$ computed by the algorithm is feasible and satisfies 
\begin{equation*}
    \bigE{f(\x)} \geq \bigl(\alpha - o(1)\bigr) \, f(\x^*),
\end{equation*}
where $\x^*$ denotes an optimal assignment that knows the functions $f$ and costs beforehand, i.e., $\x^*$ is an optimal solution for the mathematical program~\eqref{eq:main-problem}.
Note that the expectation is taken over the random permutation and over randomized decisions of the algorithm and that $o(1)$ denotes a term that vanishes for $n \to \infty$. 

\section{Submodular Generalized Assignment}
\label{sec:submodular-assignment}

In this section, we present our main result for the submodular generalized assignment problem in the random-order model. 
Apart from the budgeted submodular welfare maximization problem, the submodular generalized assignment problem also generalizes the knapsack problem with a monotone submodular objective if $ m = 1$ and, if all budgets and all costs are equal to $1$, the problem becomes the submodular secretary matching problem.



Our main result of this section is as follows.

\begin{theorem}
    \label{theo:general-beta-approximation}
    There exists an $\alpha$-competitive randomized algorithm for the submodular generalized assignment problem in the random-order model with
    \begin{equation*}
        \alpha = \beta \biggl( \frac{1}{\sqrt{e}} - \frac{1}{2} \biggr) \approx \frac{\beta}{9.39},
    \end{equation*}
where $\beta$ is the approximation for computing an offline optimal solution (or the multilinear extension thereof).
\end{theorem}

The corresponding algorithm relies on a subroutine that provides potentially fractional solutions for the offline submodular generalized assignment problem.
Without requiring a polynomial-time algorithm, we can achieve $\beta=1$ by solving the offline problem to optimality.
However, for a polynomial algorithm, we can obtain $\smash{\beta = \bigl(1- \frac{1}{e}\bigr)}$ with the continuous greedy algorithm of \citet{CalinescuCPV11}; see \Cref{lem:continuous-greedy}.

\begin{corollary}
    \label{coro:polynomial-approximation}
    The submodular generalized assignment problem in the random-order model admits
    \begin{enumerate}
        \item a polynomial $\alpha$-competitive randomized algorithm with
        $\alpha = \bigl( 1 - \frac{1}{e} \bigr) \bigl( \frac{1}{\sqrt{e}} - \frac{1}{2} \bigr) \approx \frac{1}{14.85}$, and
        \item an $\alpha$-competitive randomized algorithm  with
        $\alpha = \bigl( \frac{1}{\sqrt{e}} - \frac{1}{2} \bigr) \approx \frac{1}{9.39}$.
    \end{enumerate}
\end{corollary}

To achieve this result, we use the approach by \citet{KlimmK25} for the generalized assignment problem with linear objectives.
First, we design an algorithm that computes a potentially infeasible assignment with the property that the removal of at most one item for each agent makes the assignment feasible.
We formalize the algorithm in \Cref{alg:ro-infeasible-submodular-gap} and refer to it as \ninfeasibleSGAP.
The algorithm consists of a sampling phase and an assignment phase.
In each round~$\ell$ of the assignment phase, the algorithm first computes a fractional assignment~$\tilde{\x}(Q_\ell) $ using only the set of items $Q_\ell$ revealed up to round~$\ell$.
This fractional assignment is a $\beta$-approximation of the optimal binary assignment $\x^*(Q_\ell)$ on the revealed items, i.e., $F(\tilde{\x}(Q_\ell)) \geq \beta \, f(\x^*(Q_\ell))$. 
Then, it first randomly selects an agent based on the assignment of the current item in the fractional assignment and then it creates a tentative assignment where the current item is assigned to the randomly selected agent. Finally, the tentative assignment is realized if previous assignments do not exceed the budget of the selected agent.

\begin{algorithm}[t]
    \SetArgSty{textnormal}
    \KwIn{submodular generalized assignment instance $ (f,\c,\B)$, random permutation $\pi$}
    \vspace{-.04cm}
    \KwOut{assignment $\x$ satisfying \eqref{eq:SubmodularGapConstraint2} and \eqref{eq:SubmodularGapConstraint3}}
    $ \x \gets 0$; $Q_0 \gets \emptyset$\;
    \For{\text{rounds} $\ell \in [n]$}{
        $ Q_\ell \gets Q_{\ell-1} \cup \{\pi(\ell)\}$\;
        \If{$\ell > \lfloor n/\sqrt{e} \rfloor$}{
            $\tilde{\x}(Q_\ell) \gets $ fractional $\beta$-approximate assignment of revealed items $Q_\ell$\;
            $j^{(\ell)} \gets $ select $j \!\in\! [m] \!\cup\! \{0\}$ with $ \P{j^{(\ell)} = j} = \begin{cases}
                \tilde{x}_{\pi(\ell),j}(Q_\ell) &\text{ if } j \in [m], \\
                1 - \sum_{j' \in [m]} \tilde{x}_{\pi(\ell),j'}(Q_\ell) &\text{ if } j = 0
            \end{cases} $\;
            $\y^{(\ell)} \gets \bigl(y^{(\ell)}_{i,j}\bigr)_{i \in [n],j\in [m]}$ with $y^{(\ell)}_{i,j} = \begin{cases}
                1 &\text{ if } i=\pi(\ell), j=j^{(\ell)}, \\
                0 &\text{ otherwise},
            \end{cases} $ for all $i\in [n],j\in [m] $\;
            \If{$ j^{(\ell)} \in [m] $ \textbf{ and } $\sum_{k \in [\ell-1]} c_{\pi(k),j^{(\ell)}} x_{\pi(k),j^{(\ell)}} \leq B_{j^{(\ell)}}$}{
            $\x \gets \x + \y^{(\ell)}$\;
            }
        }
    }
    {\bfseries return} $\x$\;
    \caption{\ninfeasibleSGAP}
    \label{alg:ro-infeasible-submodular-gap}
\end{algorithm}

For the analysis of \ninfeasibleSGAP, we fix an instance of the submodular generalized assignment problem. We denote the last round of the sampling phase by~$t \in [n]$. Later optimization of this value will yield a value of $t = \lfloor n \big/ \sqrt{e}\rfloor$. 
We let $\x$ denote the final assignment of the algorithm. For each round~$\ell \in [n]$, we let $Q_\ell$ denote the set of items that are revealed in the first~$\ell$ rounds. 
For each round~$\smash{\ell > t}$, we let $\smash{\tilde{\x}(Q_\ell)}$ denote the potentially fractional $\beta$-approximate assignment of the revealed items $Q_\ell$ and we define the randomly selected agent $\smash{j^{(\ell)}}$ and the tentative assignment $\smash{\y^{(\ell)}}$ of round $\ell$ as in \Cref{alg:ro-infeasible-submodular-gap}. Note that the algorithm sets $ j^{(\ell)} = 0 $ for the event where no agent is selected.
Further, we let $\smash{\x^{(\ell)}}$ denote the assignment of a round~$\smash{\ell > t}$, i.e., $\smash{\x^{(\ell)} = \y^{(\ell)}}$ if the tentative assignment is realized and $\smash{\x^{(\ell)} = 0}$ otherwise.
Additionally, for every round $\ell>t$, we let $\smash{\x^{(\ell+)} = \sum_{k=\ell+1}^{n} \x^{(k)}}$ and $\smash{\y^{(\ell+)} = \sum_{k=\ell+1}^{n} \y^{(k)}}$ denote the sum of all realized assignments and the sum of all tentative assignments in subsequent rounds, respectively.

A crucial part of the analysis is the independence in every round $\ell>t$ between the tentative assignment and the bound on the probability that the tentative assignment is realized.
With a submodular objective function, one would usually consider the marginal increase in every round $\ell$ with respect to the assignments in the first $\ell-1$ rounds in order to obtain the total value of the assignments. This, however, clearly depends on the random order and the random decisions of the algorithm in the first $\ell-1$ rounds while we require those dependencies exclusively for the probability bound. We avoid these dependencies, as done by \citet{KesselheimT17}. They analyze similar algorithms for monotone submodular functions in settings with a cardinality constraint, the edge-weighted bipartite matching problem, and subject to linear packing constraints. For every round $\ell$ of the assignment phase, the idea is to consider the marginal increase of the tentative assignment with respect to the tentative assignments in subsequent rounds $\ell+1,\dotsc,n$, i.e., the marginal increase $f(\y^{(\ell)} \mid \y^{(\ell+)})$.

In order to see that this leads to the claimed independence, we point out that the tentative assignment of round $\ell$ is deterministically determined if we fix $\smash{Q_\ell}$, $\smash{\pi(\ell)} \in Q_\ell$, and the agent~$\smash{j^{(\ell)}}$.
Thus, for a fixed round~$\ell \in [n]$, fixing these quantities for every $k \in \{\ell, \dotsc, n\}$ determines the tentative assignments of all rounds $k \in \{\ell, \dotsc, n\}$ and therefore it also determines the marginal increase $f(\y^{(\ell)} \mid \y^{(\ell+)})$. Hence, it does not depend on the random order of the first $\ell-1$ items and not on the selected agents in rounds $t+1,\dotsc,\ell-1$. 
In the following lemma, we obtain a bound for the probability that the tentative assignment will not be realized, which depends only on the remaining quantities, thus leading to the independence.
For this and upcoming proofs, we remark that a random permutation $\pi \colon [n] \to [n]$ can be obtained by the following procedure: We draw $\pi(n)$ uniformly at random from the set of all items $Q_n = [n]$ and then, for every round $\ell$ from $n-1$ to $1$, we draw $\pi(\ell)$ uniformly at random from $Q_{\ell} = Q_{\ell+1} \setminus \{\pi(\ell+1)\}$. This shows that the item $\pi(\ell)$ can be seen as being drawn uniformly at random from the set $Q_\ell$, while $Q_\ell$ can be seen as a subset of $[n]$ with $|Q_\ell| = \ell$ drawn uniformly at random.
Note that we slightly overload the notation by using a set $Q_{\ell}$ for the items that are revealed in the first $\ell$ rounds and also for the event that those items are revealed in a random order during those rounds. 

\begin{lemma}
    \label{lem:infeasiblegap-submodular-success-probability}
    Let $\ell \in \{ t+1,\dotsc,n \}$ and let $Q_{\ell-1} \subset [n]$ with $|Q_{\ell-1}| = \ell - 1$ be arbitrary. 
    Depending on items $\pi(\ell-1),\pi(\ell-2),\dotsc,\pi(t+1)$ drawn uniformly at random from
    \begin{align*}
    Q_{\ell-1}, Q_{\ell-2} = Q_{\ell-1} \setminus \{\pi(\ell-1)\},\dotsc,Q_{t+1} = Q_{t+2} \setminus \{\pi(t+2)\},
    \end{align*}
    respectively, and depending on the random selection of an agent in rounds $t+1,\dotsc,\ell-1$, we have for every agent~$j \in [m]$ that
    \begin{equation*}
        \BiggCP{\sum_{k = t+1}^{\ell-1} c_{\pi(k),j} \, y^{(k)}_{\pi(k),j} > B_j}{Q_{\ell-1}} \leq \sum_{k = t+1}^{\ell-1} \frac{1}{k}.
    \end{equation*}
\end{lemma}

\begin{proof}
    Let $\ell \in \{ t+1,\dotsc,n \}$ and $j \in [m]$ be arbitrary and let $Q_{\ell-1}$ be as in the statement.
    We start by applying Markov's inequality on the left hand side of the claimed inequality. With linearity of expectation, we obtain that
    \begin{equation}
        \label{eq:lem-infeasiblegap-submodular}
        \BiggCP{\sum_{k = t+1}^{\ell-1} c_{\pi(k),j} \, y^{(k)}_{\pi(k),j} > B_j}{Q_{\ell-1}} \leq \frac{1}{B_j} \, \sum_{k = t+1}^{\ell-1} \BigCE{ c_{\pi(k),j} \, y^{(k)}_{\pi(k),j}}{Q_{\ell-1}}.
    \end{equation}
    For every $k \in \{t+1,\dotsc,\ell-1\}$, we let $Q_k \subseteq Q_{\ell-1}$ with $|Q_k| = k$ be arbitrary. Under the condition that the items contained in $Q_k$ are revealed in the first $k$ rounds, the item $\pi(k)$ can be seen as drawn uniformly at random from $Q_k$. Hence, we have that
    \begin{align*}
        \BigCE{c_{\pi(k),j} \, y^{(k)}_{\pi(k),j}}{Q_{k}} &= \sum_{i \in Q_k} \bigCP{\pi(k) = i}{Q_k} \, \BigCE{c_{\pi(k),j} \, y^{(k)}_{\pi(k),j}}{Q_{k}, \pi(k) = i} \\
        &= \frac{1}{k} \sum_{i \in Q_k} \BigCE{c_{\pi(k),j} \, y^{(k)}_{\pi(k),j}}{Q_{k}, \pi(k) = i}.
    \end{align*}
    The tentative assignment $\y^{(k)}$ is constructed such that $\smash{y^{(k)}_{i,j} = 1}$ if $ i=\pi(k) $ and $ j = j^{(k)} $ and $\smash{y^{(k)}_{i,j} = 0}$ otherwise. Under the conditions $Q_k$ and $ \pi(k) = i $, we have that
    \begin{equation*}
        \bigCP{j^{(k)} = j}{Q_k, \pi(k) = i} = \mathbb{E}_{\tilde{\x}}\bigl[\tilde{x}_{i,j}(Q_k)\bigr],
    \end{equation*}
    where $\tilde{\x}(Q_k)$ is a fractional $\beta$-approximate assignment of the items contained in $Q_k$ and the expectation $\mathbb{E}_{\tilde{\x}}$ is only taken over a potentially randomized computation of $\tilde{\x}(Q_k)$.
    Therefore, we obtain that
    \begin{align*}
        &\BigCE{c_{\pi(k),j} \, y^{(k)}_{\pi(k),j}}{Q_{k}, \pi(k) = i}\\
        &\qquad = \bigCP{j^{(k)} = j}{Q_{k}, \pi(k) = i} \, \BigCE{c_{\pi(k),j} \, y^{(k)}_{\pi(k),j}}{Q_{k}, \pi(k) = i, j^{(k)} = j} \\
        &\qquad = c_{i,j} \, \mathbb{E}_{\tilde{\x}}\bigl[\tilde{x}_{i,j}(Q_k)\bigr].
    \end{align*}
    By the feasibility of $\tilde{\x}(Q_k)$, we obtain for every $ k \in \{t+1,\dotsc,\ell-1\}$ that
    \begin{equation*}
        \BigCE{c_{\pi(k),j} \, y^{(k)}_{\pi(k),j}}{Q_{k}} = \frac{1}{k} \sum_{i \in Q_k} c_{i,j} \, \mathbb{E}_{\tilde{\x}}\bigl[\tilde{x}_{i,j}(Q_k)\bigr] \leq \frac{B_j}{k}.
    \end{equation*}
    As $Q_k$ is an arbitrary subset of $Q_{\ell-1}$, we conclude for every $k \in \{t+1,\dotsc,\ell-1\}$ that
    \begin{align*}
        \BigCE{c_{\pi(k),j} \, y^{(k)}_{\pi(k),j}}{Q_{\ell-1}} \leq \frac{B_j}{k},
    \end{align*}
    and together with equation~\eqref{eq:lem-infeasiblegap-submodular}, we obtain that
    \begin{equation*}
        \BiggCP{\sum_{k = t+1}^{\ell-1} c_{\pi(k),j} \, y^{(k)}_{\pi(k),j} > B_j}{Q_{\ell-1}} \leq \sum_{k=t+1}^{\ell-1} \frac{1}{k},
    \end{equation*}
    as claimed.
\end{proof}

We proceed with a lemma where we derive a bound for the marginal increase $f(\y^{(\ell)} \mid \y^{(\ell+)})$ of the tentative assignment in round~$\ell$ with respect to all tentative assignments of subsequent rounds.

\begin{lemma}
    \label{lem:infeasiblegap-submodular-tentative-values}
    Let $\ell \in \{ t+1,\dotsc,n \}$. 
    We have for the tentative assignments that
    \begin{equation*}
        \bigE{f\bigl(\y^{(\ell)} \mid \y^{(\ell+)}\bigr)} \geq \frac{\beta}{n} \, f(\x^*) - \frac{1}{\ell} \, \bigE{f\bigl(\y^{(\ell+)}\bigr)}.
    \end{equation*}
\end{lemma}

\begin{proof}
    Let $\ell \in \{ t+1,\dotsc,n \}$ be arbitrary. 
    Let $\calR$ denote the set of events where both the items $\pi(\ell+1),\dotsc,\pi(n)$ and the agents $i^{(\ell+1)},\dotsc,i^{(n)}$ are fixed. Then, for an event $R \in \calR$, the tentative assignments for all rounds from $\ell+1$ to $n$ are determined and we let $\y^{(\ell+)}(R)$ denote the sum of those assignments. Further, it determines the set of items that are revealed in the first $\ell$ rounds. We denote this set by $Q_\ell(R)$.
    Consider an arbitrary event $R \in \calR$. For a simpler notation, we omit $R$ as an argument, i.e., $\y^{(\ell+)} = \y^{(\ell+)}(R)$ and $Q_\ell = Q_\ell(R)$.
    Conditioned on $R$, the item $\pi(\ell)$ can be seen as being drawn uniformly at random from $Q_\ell$. We obtain that
    \begin{align*}
        \bigCE{f\bigl(\y^{(\ell)} \;\big\vert\; \y^{(\ell+)}\bigr)}{R} &= \sum_{i \in Q_\ell} \CP{\pi(\ell) = i}{R} \, \bigCE{f\bigl(\y^{(\ell)} \;\big\vert\; \y^{(\ell+)}\bigr)}{R, \pi(\ell) = i} \\
        &= \frac{1}{\ell} \sum_{i \in Q_\ell} \bigCE{f\bigl(\y^{(\ell)} \mid \y^{(\ell+)}\bigr)}{R, \pi(\ell) = i}.
    \end{align*}
    Further, the agent $j^{(\ell)}$ is determined based on the fractional assignment $\tilde{\x}(Q_\ell)$ where agent~$j \in [m]$ is selected with probability $\P{j^{(\ell)} = j} = \tilde{x}_{\pi(\ell),j}(Q_\ell)$. Given that $\pi(\ell)=i$ and $j^{(\ell)} = j$, we have that $\y^{(\ell)} = \mathbf{1}_{i,j}$ as the tentative assignment only assigns item~$i$ to agent~$j$. We obtain that 
    \begin{align*}
        &\sum_{i \in Q_\ell} \bigCE{f\bigl(\y^{(\ell)} \mid \y^{(\ell+)}\bigr)}{R, \pi(\ell) = i} \\
        &\qquad = \sum_{i \in Q_\ell} \sum_{j \in [m]} \bigCP{j^{(\ell)} = j}{R, \pi(\ell) = i} \, \bigCE{f\bigl(\y^{(\ell)} \mid \y^{(\ell+)}\bigr)}{R, \pi(\ell) = i,j^{(\ell)} = j} \\
        &\qquad = \sum_{i \in Q_\ell} \sum_{j \in [m]} \mathbb{E}_{\tilde{\x}}\bigl[\tilde{x}_{i,j}(Q_\ell)\bigr] \, f\bigl(\mathbf{1}_{i,j} \mid \y^{(\ell+)}\bigr),
    \end{align*}
    where we use that $\y^{(\ell+)}$ and $Q_\ell$ are determined by $R$. Note that the expectation $\mathbb{E}_{\tilde{\x}}$ is only taken over a potentially randomized computation of $\tilde{\x}(Q_\ell)$.
    Since all items $i \in Q_\ell$ are not assigned in $\y^{(\ell+)}$, we have $ \mathbf{1}_{i,j} + \y^{(\ell+)} = \mathbf{1}_{i,j} \lor \y^{(\ell+)} $ for every $i \in Q_\ell$ and $ j \in [m]$. In addition, $\y^{(\ell+)}$ is a binary assignment and hence, we obtain with definition of the multilinear extension $F$ for every $i \in Q_\ell$ and $j \in [m]$ that
    \begin{align*}
        &\mathbb{E}_{\tilde{\x}}\bigl[F\bigl(\tilde{x}_{i,j}(Q_\ell) \cdot \mathbf{1}_{i,j} \mid \y^{(\ell+)}\bigr)\bigr] \\
        &\qquad= \mathbb{E}_{\tilde{\x}}\bigl[F\bigl(\tilde{x}_{i,j}(Q_\ell) \cdot \mathbf{1}_{i,j} + \y^{(\ell+)}\bigr)\bigr] - f\bigl(\y^{(\ell+)}\bigr) \\
        &\qquad= \mathbb{E}_{\tilde{\x}}\bigl[\tilde{x}_{i,j}(Q_\ell)\bigr] \, f\bigl(\mathbf{1}_{i,j} + \y^{(\ell+)}\bigr) + \bigl(1 - \mathbb{E}_{\tilde{\x}}\bigl[\tilde{x}_{i,j}(Q_\ell)\bigr] \bigr)\, f\bigl(\y^{(\ell+)}\bigr)  - f\bigl(\y^{(\ell+)}\bigr) \\
        &\qquad= \mathbb{E}_{\tilde{\x}}\bigl[\tilde{x}_{i,j}(Q_\ell)\bigr] \, f\bigl(\mathbf{1}_{i,j} \mid \y^{(\ell+)}\bigr).
    \end{align*}
    Together with \Cref{prop:extended-submodular-inequality}, we conclude that
    \begin{align*}
        \bigCE{f\bigl(\y^{(\ell)} \;\big\vert\; \y^{(\ell+)}\bigr)}{R} &= \frac{1}{\ell} \sum_{i \in Q_\ell} \sum_{j \in [m]} \mathbb{E}_{\tilde{\x}}\bigl[F\bigl(\tilde{x}_{i,j}(Q_\ell) \cdot \mathbf{1}_{i,j} \mid \y^{(\ell+)}\bigr)\bigr] \\
        &\geq \frac{1}{\ell} \Bigl( \mathbb{E}_{\tilde{\x}}\bigl[F(\tilde{\x}(Q_\ell))\bigr] - f\bigl(\y^{(\ell+)}\bigr) \Bigr),
    \end{align*}
    where we additionally used that $ F\bigl(\y^{(\ell+)}\bigr) = f\bigl(\y^{(\ell+)}\bigr) $ since $\y^{(\ell+)}$ is a binary assignment. 

    Since the previous computations hold for every $R \in \calR$, we obtain that
    \begin{equation*}
        \bigE{f\bigl(\y^{(\ell)} \mid \y^{(\ell+)}\bigr)} \geq \frac{1}{\ell} \, \E{F(\tilde{\x}(Q_\ell))} - \frac{1}{\ell} \, \bigE{f\bigl(\y^{(\ell+)}\bigr)},
    \end{equation*}
    where $Q_\ell$ with $|Q_\ell| = \ell$ is now drawn uniformly at random from $[n]$ such that each item $i \in [n]$ is contained in $Q_\ell$ with probability $\ell / n$. To bound $\E{F(\tilde{\x}(Q_\ell))} $, we recall that $ \tilde{\x}(Q_\ell) $ is a $\beta$-approximation of the optimal integral assignment on $Q_\ell$ which we denote by $\x^*(Q_\ell)$. Further, we define for every set $Q_\ell$ the assignment $\y^*(Q_\ell) = (y^*_{i,j}(Q_\ell))_{i \in [n], j \in [m] }$ such that $ y^*_{i,j}(Q_\ell) = x^*_{i,j}$ if $i \in Q_\ell$ and $ y^*_{i,j}(Q_\ell) = 0$ if $i \notin Q_\ell$ for every $i \in [n], j \in [m]$. Note that this is the projection of the optimal integral assignment to the items contained in~$Q_\ell$.
    We obtain that
    \begin{equation*}
        \E{F(\tilde{\x}(Q_\ell))} \geq \beta \, \E{f(\x^*(Q_\ell))} \geq \beta \, \E{f(\y^*(Q_\ell))}.
    \end{equation*}
    Furthermore, we have that
    \begin{align*}
        \E{f(\y^*(Q_\ell))} &= \sum_{i \in [n]} \BiggE{ f\Biggl( \sum_{j \in [m]} y^*_{i,j}(Q_\ell) \;\Bigg\vert\; \sum_{i' \in [i-1]} \sum_{j \in [m]} y^*_{i',j}(Q_\ell) \Biggr)} \\
        &\geq \sum_{i \in [n]} \BiggE{ f\Biggl( \sum_{j \in [m]} y^*_{i,j}(Q_\ell) \;\Bigg\vert\; \sum_{i' \in [i-1]} \sum_{j \in [m]} x^*_{i',j} \Biggr)} \\
        &= \sum_{i \in [n]} \P{i \in Q_\ell} \, f\Biggl( \sum_{j \in [m]} x^*_{i,j} \;\Bigg\vert\; \sum_{i' \in [i-1]} \sum_{j \in [m]} x^*_{i',j} \Biggr) \\
        &= \frac{\ell}{n} \, f(\x^*),
    \end{align*}
    where we use linearity of expectation after splitting the total value into the marginal increases for the first equality.
    The inequality holds since the marginal increases of a submodular function are monotonically decreasing, which applies since $y^*_{i,j}(Q_\ell) = 1$ implies that $x^*_{i,j} = 1$ for every $Q_\ell$ and every $i \in [n], j \in [m]$.
    For the following equation, we use that we have $\smash{y^*_{i,j}(Q_\ell) = x^*_{i,j}}$ under the condition that~$i \in Q_\ell$ and finally, we used $ \P{i \in Q_\ell} = \ell / n$ for every item~$i \in [n]$.
    Combining everything shows that
    \begin{equation*}
        \bigE{f\bigl(\y^{(\ell)} \mid \y^{(\ell+)}\bigr)} \geq \frac{\beta}{n} \, f(\x^*) - \frac{1}{\ell} \, \bigE{f\bigl(\y^{(\ell+)}\bigr)},
    \end{equation*}
    as claimed in the statement of the lemma.   
\end{proof}

With the previous two lemmas, we are now ready to bound the total value of the assignment computed by \Cref{alg:ro-infeasible-submodular-gap}.

\begin{lemma}
    \label{lem:infeasiblegap-submodular-total-value}
    Let $\x$ be the infeasible assignment computed by \Cref{alg:ro-infeasible-submodular-gap}. Then, we have that
    \begin{equation*}
        \E{f(\x)} \geq \beta \,\biggl( \frac{2}{\sqrt{e}} - 1 - o(1) \biggr)  \, f(\x^*).
    \end{equation*}
\end{lemma}

\begin{proof}
    Let $t \geq n/2$ and let $\x$ be as in the statement.
    By splitting the value of the assignment into the marginal increases of every round with respect to subsequent rounds, we obtain that
    \begin{equation*}
        \E{f(\x)} = \sum_{\ell=t+1}^n \bigE{f\bigl(\x^{(\ell)} \mid \x^{(\ell+)} \bigr)} \geq \sum_{\ell=t+1}^n \bigE{f\bigl(\x^{(\ell)} \mid \y^{(\ell+)} \bigr)},
    \end{equation*}
    where the inequality follows from submodularity of $f$ since $ \y^{(\ell+)} \geq \x^{(\ell+)} $. 
    For the right hand side, we observe for every round $ \ell > t$ that $ \x^{(\ell)} = \y^{(\ell)} $ if the tentative assignment is realized and $\x^{(\ell)} = 0 $ otherwise. Hence, $ f(\x^{(\ell)} \mid \y^{(\ell+)} ) = f(\y^{(\ell)} \mid \y^{(\ell+)} )$ if the tentative assignment is realized and $f(\x^{(\ell)} \mid \y^{(\ell+)}) = 0$ otherwise.  
    Recall that the tentative assignment of round $\ell$ depends on $Q_\ell$, $\pi(\ell)$ and $ j^{(\ell)} $. Then, $ f(\y^{(\ell)} \mid \y^{(\ell+)}) $ is fixed if we determine the items $\pi(n),\dotsc,\pi(\ell)$ and the agents $ j^{(n)},\dotsc,j^{(\ell)}$. In particular, this defines the set $Q_{\ell-1}$ of items revealed in the first $\ell-1$ rounds and the agent $ j^{(\ell)} $ that gets the item in round $\ell$.
    \Cref{lem:infeasiblegap-submodular-success-probability} shows for every set $Q_{\ell-1}$ and every agent $j \in [m]$ that the total cost of the assignments from the first $\ell-1$ rounds exceeds the budget of agent~$j$ with a probability of at most $\smash{\sum_{k = t+1}^{\ell-1} \frac{1}{k}}$.
    Thus, independent of the marginal increase in round $\ell$, the tentative assignment gets realized with a probability of at least $\smash{1 - \sum_{k=t+1}^{\ell-1} \frac{1}{k}}$. Hence, we obtain that
    \begin{equation*}
        \sum_{\ell=t+1}^n \bigE{f\bigl(\x^{(\ell)} \mid \y^{(\ell+)}\bigr)} \geq \sum_{\ell=t+1}^{n} \Biggl(1 - \sum_{k = t+1}^{\ell-1} \frac{1}{k}\Biggr) \, \bigE{f\bigl(\y^{(\ell)} \mid \y^{(\ell+)}\bigr)}.
    \end{equation*}
    Note that the bound of \Cref{lem:infeasiblegap-submodular-tentative-values} for $ \E{f(\y^{(\ell)} \mid \y^{(\ell+)} )} $ depends on $\E{f(\y^{(\ell+)})}$ for every $\ell \in \{t+1,\dotsc,n\}$. We can split the latter into the sum of marginals $\E{f(\y^{(k)} \mid \y^{(k+)} )}$ for all $k > \ell$. Therefore, we first apply \Cref{lem:infeasiblegap-submodular-tentative-values} for $\ell=t+1$, simplify the remaining marginals, apply the lemma for the next round, and so on. We claim that this leads to 
    \begin{equation}
        \label{eq:lem-infeasiblegap-submodular-total-value}
        \sum_{\ell=t+1}^{n} \Biggl(1 - \sum_{k = t+1}^{\ell-1} \frac{1}{k}\Biggr) \, \bigE{f\bigl(\y^{(\ell)} \mid \y^{(\ell+)}\bigr)} \geq \frac{\beta}{n} \sum_{\ell=t+1}^{n} \Biggl( \frac{2t}{\ell-1} - 1 \Biggr) \, f(\x^*).
    \end{equation}
    In order to prove this inequality, we show for every $h \in \{t,\dotsc,n\}$ that
    \begin{multline*}
        \sum_{\ell=t+1}^{n} \Biggl(1 - \sum_{k = t+1}^{\ell-1} \frac{1}{k}\Biggr) \, \bigE{f\bigl(\y^{(\ell)} \mid \y^{(\ell+)}\bigr)} \\
        \geq \sum_{\ell=h+1}^n \Biggl( \frac{2t}{h} - 1 - \sum_{k=h+1}^{\ell-1} \frac{1}{k} \Biggr) \, \bigE{f\bigl(\y^{(\ell)} \mid \y^{(\ell+)}\bigr)} + \frac{\beta}{n} \sum_{\ell=t+1}^{h} \Biggl( \frac{2t}{\ell-1} - 1 \Biggr) \, f(\x^*).
    \end{multline*}
    For $h = n$ this implies the claimed inequality. Proving this by induction, we first observe that the right hand side is equal to the left hand side for $h = t$. Assuming that this inequality holds for a fixed $h \in \{t,\dotsc,n-1\}$, it remains to show that
    \begin{multline*}
        \sum_{\ell=h+1}^n \Biggl( \frac{2t}{h} - 1 - \sum_{k=h+1}^{\ell-1} \frac{1}{k} \Biggr) \, \bigE{f\bigl(\y^{(\ell)} \mid \y^{(\ell+)} \bigr)} + \frac{\beta}{n} \sum_{\ell=t+1}^{h} \Biggl( \frac{2t}{\ell-1} - 1 \Biggr) \, f(\x^*) \\
        \geq \sum_{\ell=h+2}^n \Biggl( \frac{2t}{h+1} - 1 - \sum_{k=h+2}^{\ell-1} \frac{1}{k} \Biggr) \, \bigE{f\bigl(\y^{(\ell)} \mid \y^{(\ell+)} \bigr)} + \frac{\beta}{n} \sum_{\ell=t+1}^{h+1} \Biggl( \frac{2t}{\ell-1} - 1 \Biggr) \, f(\x^*).
    \end{multline*}
    To apply \Cref{lem:infeasiblegap-submodular-tentative-values} for $\ell=h+1$ on the left hand side of this inequality, we first check that the factor $\smash{\frac{2t}{h} - 1 - \sum_{k=h+1}^{h} \frac{1}{k} = \frac{2t}{h} - 1}$ is non-negative. This is the case by our assumption that $t \geq n/2$, which implies $2t \geq n \geq h$. After applying the lemma, we obtain
    {\allowdisplaybreaks
    \begin{align*}
        &\sum_{\ell=h+1}^n \Biggl( \frac{2t}{h} - 1 - \sum_{k=h+1}^{\ell-1} \frac{1}{k} \Biggr) \, \bigE{f\bigl(\y^{(\ell)} \mid \y^{(\ell+)} \bigr)} + \frac{\beta}{n} \sum_{\ell=t+1}^{h} \Biggl( \frac{2t}{\ell-1} - 1 \Biggr) \, f(\x^*) \\
            &\quad\geq \sum_{\ell=h+2}^n \Biggl( \frac{2t}{h} - 1 - \sum_{k=h+1}^{\ell-1} \frac{1}{k} \Biggr) \, \bigE{f\bigl(\y^{(\ell)} \mid \y^{(\ell+)}\bigr)} + \frac{\beta}{n} \sum_{\ell=t+1}^{h} \Biggl( \frac{2t}{\ell-1} - 1 \Biggr) \, f(\x^*) \\
            &\qquad + \Biggl( \frac{2t}{h} - 1 \Biggr) \, \Biggl(\frac{\beta}{n} \, f(\x^*) - \frac{1}{h+1} \, \bigE{f\bigl(\y^{((h+1)+)}\bigr)}\Biggr) \\
            &\quad= \sum_{\ell=h+2}^n \Biggl( \frac{2t}{h} \biggl(1 - \frac{1}{h+1}\biggr) + \frac{1}{h+1} - 1 - \sum_{k=h+1}^{\ell-1} \frac{1}{k} \Biggr) \, \bigE{f\bigl(\y^{(\ell)} \mid \y^{(\ell+)}\bigr)}  + \frac{\beta}{n} \sum_{\ell=t+1}^{h+1} \Biggl( \frac{2t}{\ell-1} - 1 \Biggr) \, f(\x^*) \\
            &\quad = \sum_{\ell=h+2}^n \Biggl( \frac{2t}{h+1} - 1 - \sum_{k=h+2}^{\ell-1} \frac{1}{k} \Biggr) \, \bigE{f\bigl(\y^{(\ell)} \mid \y^{(\ell+)}\bigr)} + \frac{\beta}{n} \sum_{\ell=t+1}^{h+1} \Biggl( \frac{2t}{\ell-1} - 1 \Biggr) \, f(\x^*),
    \end{align*}
    }%
    where we used for the first equation that $ \E{f(\y^{((h+1)+)})} = \sum_{\ell=h+2}^n \E{f(\y^{(\ell)} \mid \y^{(\ell+)} )} $.
    We conclude that equation~\eqref{eq:lem-infeasiblegap-submodular-total-value} holds and that we have
    \begin{equation*}
        \E{f(\x)} \geq \frac{\beta}{n} \sum_{\ell=t+1}^{n} \Biggl( \frac{2t}{\ell-1} - 1 \Biggr) \, f(\x^*).
    \end{equation*}
    Since we have that
    \begin{equation*}
        \sum_{\ell=t+1}^{n} \frac{1}{\ell-1} = \sum_{\ell=t}^{n-1} \frac{1}{\ell} \geq \int_{\ell=t}^n \frac{1}{\ell} \, \d\ell = \ln\biggl(\frac{n}{t}\biggr),
    \end{equation*}
    we obtain
    \begin{equation*}
        \E{f(\x)} \geq \beta \, \biggl(\frac{2t}{n} \, \ln\biggl(\frac{n}{t}\biggr) - \frac{n-t}{n} \biggr) \, f(\x^*) = \beta \, \biggl(\frac{t}{n} \, \biggl(2 \, \ln\biggl(\frac{n}{t}\biggr) + 1 \biggr) - 1 \biggr) \, f(\x^*).
    \end{equation*}
    Optimization over $\frac{t}{n} \in \bigl[\frac{1}{2},1\bigr]$ yields $\frac{t}{n} = \frac{1}{\sqrt{e}}$ and
    \begin{equation*}
        \E{f(\x)} \geq \beta \, \biggl(\frac{2}{\sqrt{e}} - 1 \biggr) \, f(\x^*).
    \end{equation*}
    With $t = \lfloor n / \sqrt{e} \rfloor$, we obtain the statement of the lemma.
\end{proof}

\begin{algorithm}[t]
    \SetArgSty{textnormal}
    \KwIn{submodular generalized assignment instance $ (f,\c,\B)$, random permutation $\pi$}
    \vspace{-.04cm}
    \KwOut{assignment $\x$ satisfying \eqref{eq:SubmodularGapConstraint1}, \eqref{eq:SubmodularGapConstraint2}, and \eqref{eq:SubmodularGapConstraint3}}
    $X \gets \text{Bernoulli}(1/2)$\;
    $\x \gets 0$; $Q_0 \gets \emptyset$\;
    $\y \gets 0$\;
    \For{rounds $\ell \in [n]$}{
        $ Q_\ell \gets Q_{\ell-1} \cup \{\pi(\ell)\}$\;
        \If{$\ell > \lfloor n/\sqrt{e} \rfloor$}{
            $\tilde{\x}(Q_\ell) \gets $ fractional $\beta$-approximate assignment of revealed items $Q_\ell$\;
            $j^{(\ell)} \gets $ select $j \!\in\! [m] \!\cup\! \{0\}$ with $ \P{j^{(\ell)} \!=\! j} \!=\! \begin{cases}
                \tilde{x}_{\pi(\ell),j}(Q_\ell) &\!\!\text{if } j \in [m] \\
                1 \!-\! \sum_{j' \!\in\! [m]} \tilde{x}_{\pi(\ell),j'} (Q_\ell) &\!\!\text{if } j \!=\! 0
            \end{cases} $\;
            $\y^{(\ell)} \gets \bigl(y^{(\ell)}_{i,j}\bigr)_{i \in [n],j\in [m]}$ with $y^{(\ell)}_{i,j} = \begin{cases}
                1 &\text{ if } i = \pi(\ell), j=j^{(\ell)}, \\
                0 &\text{ otherwise},
            \end{cases} $ for all $ i \in [n], j\in [m] $\;
            \If{$ j^{(\ell)} \in [m] $ \textbf{ and } $\sum_{k \in [\ell-1]} c_{\pi(k),j^{(\ell)}} ( x_{\pi(k),j^{(\ell)}} + y_{\pi(k),j^{(\ell)}}) \leq B_{j^{(\ell)}}$}{
                \eIf{$X = 1$}{
                    \eIf{$ c_{\pi(\ell),j^{(\ell)}} + \sum_{k \in [\ell-1]} c_{\pi(k),j^{(\ell)}} x_{\pi(k),j^{(\ell)}} \leq B_{j^{(\ell)}} $}{
                        $\x \gets \x + \y^{(\ell)}$\;
                    }{
                        $\y \gets \y + \y^{(\ell)}$\;
                    }
                }{
                    \eIf{$ c_{\pi(\ell),j^{(\ell)}} + \sum_{k \in [\ell-1]} c_{\pi(k),j^{(\ell)}} y_{\pi(k),j^{(\ell)}} \leq B_{j^{(\ell)}} $}{
                        $\y \gets \y + \y^{(\ell)}$\;
                    }{
                        $\x \gets \x + \y^{(\ell)}$\;
                    }
                }
            }
        }
    }
    {\bfseries return} $\x$\;
    \caption{\nfeasibleSGAP}
    \label{alg:ro-feasible-submodular-gap}
\end{algorithm}

Now we transform the infeasible assignment of \ninfeasibleSGAP\ into two feasible assignments.
As in \citet{KlimmK25}, the first assignment contains all the assignments of \ninfeasibleSGAP\ that can be made without violation of a budget constraint and the second assignment contains all the assignments that violate a budget constraint in \ninfeasibleSGAP.
We formalize this procedure in \Cref{alg:ro-feasible-submodular-gap} and refer to the algorithm as \nfeasibleSGAP. As we lose an additional factor of $1/2$ by randomizing over both assignments, we obtain the result of \Cref{theo:general-beta-approximation}.

\begin{proof}[Proof of \Cref{theo:general-beta-approximation}]
    We claim the result for \nfeasibleSGAP. Consider an arbitrary instance of the submodular generalized assignment problem. As in \Cref{alg:ro-feasible-submodular-gap}, we let $\x$ denote the assignment of \nfeasibleSGAP\ and we let $\y$ denote the auxiliary assignment. Further, let $\hat{\x}$ denote the infeasible assignment returned by \ninfeasibleSGAP. If we consider a fixed permutation $\pi \in \Pi$, fixed agents $j^{(\ell)}$ and fixed assignments $ \tilde{\x}(Q_\ell)$ for every round $\ell \in [n]$ that belongs to the assignment phase, and a fixed value for the Bernoulli distributed random variable $X$, then we have that $\x + \y = \hat{\x}$. Due to the submodularity of $f$, this implies that $f(\x) + f(\y) \geq f(\hat{\x})$. Now taking the expectation over the random order, the randomized selections of the agents, and the potentially randomized assignments $\tilde{\x}$ of every round gives $ \E{f(\x)} + \E{f(\y)} \geq \E{f(\hat{\x})}$. Taking the Bernoulli distributed random variable $X$ into account implies that $ \E{f(\x)} = \E{f(\y)} $ since we have by the construction of the algorithm that $\CE{f(\x)}{X=1} = \CE{f(\y)}{X=0}$ and $\CE{f(\x)}{X=0} = \CE{f(\y)}{X=1}$.
    We conclude with \Cref{lem:infeasiblegap-submodular-total-value} that
    \begin{align*}
        \E{f(\x)} = \frac{1}{2} \bigl( \E{f(\x)} + \E{f(\y)} \bigr) \geq \frac{1}{2} \, \E{f(\hat{\x})} &\geq \beta \biggl( \frac{1}{\sqrt{e}} - \frac{1}{2} - o(1) \biggr)  \, f(\x^*),
    \end{align*}
    This concludes the proof of the theorem.
\end{proof}

A standard large markets assumption is that each assignment contributes only a negligible fraction to the optimal solution, i.e., $m\,f(\mathbf{1}_{i,j}) / f(\x^*) \in o_n(1)$ for all items~$i$ and agents~$j$.
Under this assumption, we can afford to lose the~$m$ items last assigned to each agent that potentially make the assignment infeasible, and run only the part of the algorithm that assigns all items as long as the budget constraints of each agent are not violated.
Compared to the competitive ratios of \Cref{coro:polynomial-approximation}, we gain a factor of $2$ and obtain the following results.

\begin{corollary}
    \label{coro:large-markets}
    For the submodular generalized assignment problem in the random-order model and large markets, the following holds:
\begin{enumerate}
    \item There exists a polynomial $\alpha$-competitive randomized algorithm with
    $\alpha = \bigl( 1 \!-\! \frac{1}{e} \bigr) \bigl( \frac{2}{\sqrt{e}} \!-\! 1 \bigr) \approx \frac{1}{7.42}$.
    \item There exists an $\alpha$-competitive randomized algorithm  with
    $\alpha = \bigl( \frac{2}{\sqrt{e}} \!-\! 1 \bigr) \approx \frac{1}{4.69}$.
\end{enumerate}
\end{corollary}

\section{Submodular Secretary Matching}
\label{sec:bipartite-matching}

We proceed in this section with the special case of the submodular generalized assignment problem where all costs and budgets are equal to $1$. This problem corresponds to the submodular secretary matching problem where the weight of a set of edges is given by a monotone submodular function. 
Formally, we have given a bipartite graph $G = (U \cup V, E)$ with $|U| = n$ and a normalized monotone submodular function $f \colon 2^E \to \R$. Note that this problem also contains the problem of maximizing a monotone submodular function subject to a transversal matroid.
In the random-order model, an algorithm initially knows the cardinality of $U$ and the vertices from $V$. Then, the vertices of~$U$ are revealed to the algorithm in a random order given by~$\pi \in \Pi$. 
In every round $\ell \in [n]$, a vertex~$\pi(\ell) \in U$ is revealed, and the algorithm observes its incident edges and may query the function~$f$ for sets $A \subseteq \{ (u,v) \in \{\pi(1),\dots,\pi(\ell) \} \times V \}$.
The algorithm then has to decide immediately and irrevocably whether to select an edge that matches $\pi(\ell)$ to a vertex in $V$, or to leave the vertex $\pi(\ell)$ unmatched. 


We approach the problem similarly as in \Cref{sec:submodular-assignment}. Compared to the submodular generalized assignment problem, however, this special case does not require infeasible assignments. Therefore, we do not need to split an algorithm into two feasible algorithms, which already improves the competitiveness by a factor of $2$ to $ \smash{\beta \bigl(\frac{2}{\sqrt{e}} - 1\bigr) \approx \frac{\beta}{4.69}}$.
We show that the competitiveness improves even more, as we obtain a better bound on the probability for the realization of a tentative assignment. This then leads to the following result:

\begin{theorem}
    \label{theo:bipartite-matching}
    There exists a $\alpha$-competitive randomized algorithm for the submodular secretary matching problem in the random-order model with
    \begin{equation*}
        \alpha = \beta \bigl( \sqrt{2}-1 \bigr) e^{\sqrt{2}-2} \approx \frac{\beta}{4.34},
    \end{equation*}
where $\beta$ is the approximation for computing an offline optimal solution (or the multilinear extension thereof).
\end{theorem}

Analogously to \Cref{coro:polynomial-approximation}, we can use the $\smash{\bigl(1-\frac{1}{e}\bigr)}$-approximation for the multilinear extension of \citet{CalinescuCPV11} to obtain a polynomial-time algorithm, and solve the offline problem optimally when super-polynomial runtime is allowed.
This yields the following immediate corollary.

\begin{corollary}
    \label{coro:polynomial-approximation-matching}
    The submodular secretary matching problem in the random-order model admits
\begin{enumerate}
    \item a polynomial $\alpha$-competitive randomized algorithm with
    $\alpha = \bigl( 1 \!-\! \frac{1}{e} \bigr) \bigl( \sqrt{2} \!-\! 1 \bigr)e^{\sqrt{2} - 2} \approx \frac{1}{6.86}$, and
    \item an $\alpha$-competitive randomized algorithm  with
    $\alpha = \bigl( \sqrt{2} \!-\! 1 \bigr)e^{\sqrt{2} - 2} \approx \frac{1}{4.34}$.
\end{enumerate}
    
\end{corollary}

We formalize the algorithm for the bipartite matching problem in \Cref{alg:ro-bipartite-matching}. 
For the analysis, we fix an instance of the submodular secretary matching problem. 
As in the previous section, we denote the last round of the sampling phase by~$t \in [n]$, the set of vertices that are revealed in the first~$\ell$ rounds by $Q_\ell$, and we let $M$ and $M^*$ denote the matching returned by the algorithm and an optimal matching, respectively. 
In each round~$\smash{\ell > t}$, the algorithm computes a potentially fractional $\beta$-approximate matching on $Q_\ell \cup V$ that we denote by $\smash{\tilde{\x}(Q_\ell)}$.
As in \Cref{alg:ro-bipartite-matching}, we define the randomly selected vertex $\smash{v^{(\ell)} \in V}$ and we define $\smash{e^{(\ell)} = ( \pi(\ell), v^{(\ell)} )} $ if a vertex is selected in round~$\ell$.
Further, we let $\smash{M^{(\ell)}}$ denote the set that contains the edge added to $M$ in a round~$\smash{\ell > t}$, i.e., $\smash{M^{(\ell)} = \{e^{(\ell)}\} }$ if the tentative edge is added to $M$ and $\smash{M^{(\ell)} = \emptyset}$ otherwise.
Additionally, for every round $\ell>t$, we let $\smash{M^{(\ell+)} = \bigcup_{k=\ell+1}^{n} M^{(k)}}$ and $\smash{E^{(\ell+)} = \bigcup_{k=\ell+1}^{n} \{e^{(k)}\}}$ denote the union of all edges added to $M$ and the union of all tentative edges in subsequent rounds, respectively.

\begin{algorithm}[t]
    \SetArgSty{textnormal}
    \KwIn{bipartite graph $G = (U \cup V, E)$, set function $f \colon 2^E \to \R$, random permutation $\pi$}
    \vspace{-.04cm}
    \KwOut{matching $M \subseteq E$}
    $ M \gets \emptyset$; $Q_0 \gets \emptyset$\;
    \For{\text{rounds} $\ell \in [n]$}{
        $ Q_\ell \gets Q_{\ell-1} \cup \{\pi(\ell)\}$\;
        \If{$\ell > \big\lfloor n/e^{2-\sqrt{2}} \big\rfloor$}{
            $\tilde{\x}(Q_\ell) \gets $ fractional $\beta$-approximate matching on $Q_\ell \cup V$\;
            $v^{(\ell)} \gets $ randomly select $v \in V \, \dot\cup\, \{0\}$ where $ \P{v^{(\ell)} = v} = \begin{cases}
                \tilde{x}_{\pi(\ell),v} (Q_\ell) &\text{ if } v \in V \\
                1 - \sum_{v' \in V} \tilde{x}_{\pi(\ell),v'} (Q_\ell) &\text{ if } v=0
            \end{cases} $\;
            \If{$v^{(\ell)} \in V$}{
                $e^{(\ell)} = ( \pi(\ell), v^{(\ell)} ) $\;
                \If{$ v^{(\ell)} $ is not covered by $M$}{
                    $M \gets M \cup \{e^{(\ell)}\}$\;
                }
            }
        }
    }
    {\bfseries return} $M$\;
    \caption{\textsc{SubMatching}}
    \label{alg:ro-bipartite-matching}
\end{algorithm}

The additional improvement of the competitive ratio relies on an improved probability for the event that the tentative edge is added to the matching compared to the bound that is implied by \Cref{lem:infeasiblegap-submodular-success-probability}. Note that the following lemma generalizes a part of \citet[Lemma~1]{KesselheimRTV13} such that we can use a fractional $\beta$-approximate matching in every round of the algorithm that we can obtain by the continuous greedy algorithm with $ \smash{\beta = 1 - \frac{1}{e}}$.

\begin{lemma}
    \label{lem:bipartite-matching-success-probability}
    Let $\ell \in \{ t+1,\dotsc,n \}$ and let $Q_{\ell-1} \subset [n]$ with $|Q_{\ell-1}| = \ell - 1$ be arbitrary. 
    Depending on items $\pi(\ell-1),\pi(\ell-2),\dotsc,\pi(t+1)$ drawn uniformly at random from
    \begin{align*}
        Q_{\ell-1}, Q_{\ell-2} = Q_{\ell-1} \setminus \{\pi(\ell-1)\},\dotsc,Q_{t+1} = Q_{t+2} \setminus \{\pi(t+2)\},
    \end{align*}
    respectively, and depending on the randomized selection of a vertex from $V$ in rounds $t+1,\dotsc,\ell-1$, we have for every $v \in V$ that
    \begin{equation*}
        \bigCP{v \text{ is not covered by $M$ after round $\ell-1$}}{Q_{\ell-1}} \geq \prod_{k=t+1}^{\ell-1} \biggl(1 - \frac{1}{k}\biggr).
    \end{equation*}
\end{lemma}

\begin{proof}
    Let $v \in V$ and $\ell \in \{ t+1,\dotsc,n \}$ be arbitrary and let $ Q_{\ell-1} $ be as in the statement.
    The vertex $v$ is not covered by the matching $M$ after round $\ell-1$ if, in every round $k \in \{t+1,\dotsc,\ell-1\}$, it is not selected as the endpoint $v^{(k)}$ of the tentative edge $e^{(k)}$. 
    We show in the following for every $k \in \{t+1,\dotsc,\ell-1\}$ that $ v^{(k)} = v$ has a probability of at most $1/k$ and, most importantly, that this bound only depends on the randomized selection of $\pi(k)$ and $v^{(k)}$ in round $k$.
    Therefore, let $k \in \{t+1,\dotsc,\ell-1\}$ and $Q_k \subseteq Q_{\ell-1}$ with $|Q_k| = k$ be arbitrary.
    Given $Q_k$, the item $\pi(k)$ is uniformly drawn from $Q_k$. We obtain that
    \begin{equation*}
        \bigCP{v^{(k)} = v}{Q_k} = \frac{1}{k} \sum_{u \in Q_k} \bigCP{v^{(k)} = v}{Q_k, \pi(k) = u}.
    \end{equation*}
    Now, for each $u \in Q_k$, the probability that $v^{(k)} = v$ under the condition that $\pi(k) = u$ is given by $\tilde{x}_{u,v}(Q_k)$. Note that $\tilde{\x}(Q_k)$ might be randomized. Due to the feasibility of the fractional matching, we know that $ \sum_{u \in Q_k} \tilde{x}_{u,v}(Q_k) \leq 1 $ for every realization of $\tilde{\x}(Q_k)$ and this leads to 
    \begin{equation*}
        \frac{1}{k} \sum_{u \in Q_k} \bigCP{v^{(k)} = v}{Q_k, \pi(k) = u} = \frac{1}{k} \sum_{u \in Q_k} \mathbb{E}_{\tilde{\x}}\bigl[\tilde{x}_{u,v}(Q_k)\bigr] \leq \frac{1}{k},
    \end{equation*}
    where the expectation $\mathbb{E}_{\tilde{\x}}$ is only taken over the potentially randomized computation of~$\tilde{\x}(Q_k)$.
    As we used independent events for every $k \in \{t+1,\dotsc,\ell-1\}$, we conclude that
    \begin{align*}
        &\bigCP{v \text{ is not covered by $M$ after round $\ell-1$}}{Q_{\ell-1}} \\
        &\qquad= \bigCP{v^{(k)} \neq v \text{ for all } k \in \{t+1,\dotsc,\ell-1\}}{Q_{\ell-1}} \geq \prod_{k=t+1}^{\ell-1} \biggl(1 - \frac{1}{k}\biggr),
    \end{align*}
    as claimed.
\end{proof}

With the previous lemma and the bound on the tentative values from \Cref{sec:submodular-assignment}, we now proceed with the proof of \Cref{theo:bipartite-matching} where we bound the expected value of the matching returned by \Cref{alg:ro-bipartite-matching}. 
As in the proof of \Cref{theo:general-beta-approximation}, we inductively apply \Cref{lem:infeasiblegap-submodular-tentative-values}, which again leads to a tedious computation. Also the optimization of the resulting bound requires a more elaborate computation. Those parts of the proof are deferred to \Cref{app:deferred-proofs} in the appendix.

\begin{proof}[Proof of \Cref{theo:bipartite-matching}]
    Let $t \geq n/2$ and let $M$ be the matching returned by \Cref{alg:ro-bipartite-matching}. We divide the total value of the edges contained in $M$ into the marginal increases of every round with respect to subsequent rounds and we use that the selected edges are a subset of the tentative edges. Then, we obtain that
    \begin{equation*}
        \E{f(M)} = \sum_{\ell=t+1}^n \bigE{f\bigl(M^{(\ell)} \mid M^{(\ell+)} \bigr)} \geq \sum_{\ell=t+1}^n \bigE{f\bigl(M^{(\ell)} \mid E^{(\ell+)} \bigr)}.
    \end{equation*}
    In every round $ \ell > t$, we have that $ M^{(\ell)} = \{e^{(\ell)}\} $ if a vertex $v^{(\ell)} \in V$ is selected and the tentative edge is added to $M$, otherwise, we have $M^{(\ell)} = \emptyset $. 
    Hence, $ f(M^{(\ell)} \mid E^{(\ell+)} ) = f(e^{(\ell)} \mid E^{(\ell+)} )$ if the tentative assignment is realized and $f(M^{(\ell)} \mid E^{(\ell+)}) = 0$ otherwise.  
    The tentative edge $e^{(\ell)}$ only depends on $Q_\ell$, $\pi(\ell)$, and $v^{(\ell)}$.
    Thus, $ f(e^{(\ell)} \mid E^{(\ell+)}) $ is fixed if we fix the online vertices $\pi(n),\dotsc,\pi(\ell)$ and the offline vertices $v^{(n)},\dotsc,v^{(\ell)}$.
    This defines the set $Q_{\ell-1}$ of vertices revealed in the first $\ell-1$ rounds. \Cref{lem:bipartite-matching-success-probability} shows that the probability that the endpoint of the tentative edge $e^{(\ell)}$ in $V$ was not matched in an earlier round is at least 
    \begin{equation*}
        \prod_{k=t+1}^{\ell-1} \biggl(1 - \frac{1}{k}\biggr) = \frac{t}{\ell-1}.
    \end{equation*}
    In particular, this bound only depends on the random order and the selected offline vertices of the first $\ell-1$ rounds and thus, it does not depend on the set $Q_{\ell-1}$ or the vertices $v^{(n)},\dotsc,v^{(\ell)}$. 
    Hence, we obtain that
    \begin{equation*}
        \sum_{\ell=t+1}^n \bigE{f\bigl(M^{(\ell)} \mid E^{(\ell+)}\bigr)} \geq \sum_{\ell=t+1}^{n} \frac{t}{\ell-1} \, \bigE{f\bigl(e^{(\ell)} \mid E^{(\ell+)}\bigr)}.
    \end{equation*}
    Now, using that the submodular secretary matching problem is a special case of the problem studied in \Cref{sec:submodular-assignment}, we can apply \Cref{lem:infeasiblegap-submodular-tentative-values} that states for every $\ell \in \{ t+1,\dotsc,n \}$ that
    \begin{equation*}
        \bigE{f\bigl(e^{(\ell)} \mid E^{(\ell+)}\bigr)} \geq \frac{\beta}{n} \, f(M^*) - \frac{1}{\ell} \, \bigE{f\bigl(E^{(\ell+)}\bigr)}.
    \end{equation*}
    We claim that this leads to the following inequality:
    \begin{equation}
        \sum_{\ell=t+1}^{n} \frac{t}{\ell-1} \, \bigE{f\bigl(e^{(\ell)} \mid E^{(\ell+)}\bigr)} \geq \beta \frac{t}{n} \sum_{\ell=t+1}^{n} \frac{1}{\ell-1} \biggl( 1 - \sum_{k=t+1}^{\ell-1} \frac{1}{k-1} \biggr) \, f(M^*).
    \end{equation}
    The proof of this inequality is deferred to \Cref{lem:appendix-induction-inequality} in \Cref{app:deferred-proofs} in the appendix. 
    We conclude that we have
    \begin{equation*}
        \E{f(M)} \geq \beta \frac{t}{n} \sum_{\ell=t+1}^{n} \frac{1}{\ell-1} \biggl( 1 - \sum_{k=t+1}^{\ell-1} \frac{1}{k-1} \biggr) \, f(M^*).
    \end{equation*}
    For the statement of the theorem it now remains to show that
    \begin{equation*}
        \frac{t}{n} \sum_{\ell=t+1}^{n} \frac{1}{\ell-1} \biggl( 1 - \sum_{k=t+1}^{\ell-1} \frac{1}{k-1} \biggr) \geq e^{\sqrt{2}-2}\bigl(\sqrt{2}-1\bigr) + o(1).
    \end{equation*}
    This part of the proof is deferred to \Cref{lem:appendix-optimization} in \Cref{app:deferred-proofs} in the appendix. 
\end{proof}

\newpage
\appendix
\section{Deferred Parts from the Proof of \texorpdfstring{\Cref{theo:bipartite-matching}}{Theorem~4.1}}
\label{app:deferred-proofs}

\begin{lemma}
    \label{lem:appendix-induction-inequality}
    We have that
    \begin{equation*}
        \sum_{\ell=t+1}^{n} \frac{t}{\ell-1} \, \bigE{f\bigl(e^{(\ell)} \mid E^{(\ell+)}\bigr)} \geq \beta \frac{t}{n} \sum_{\ell=t+1}^{n} \frac{1}{\ell-1} \biggl( 1 - \sum_{k=t+1}^{\ell-1} \frac{1}{k-1} \biggr) \, f(M^*).
    \end{equation*}
\end{lemma}

\begin{proof}
    For ease of exposition, we define $ x^{(\ell)} = \E{f(e^{(\ell)} \mid E^{(\ell+)})} $. Note that for every $ \ell \in \{ t+1,\dotsc,n \}$ the statement of \Cref{lem:infeasiblegap-submodular-tentative-values} becomes 
    \begin{equation}
        \label{eq:theo-bipartite-matching-restated-lemma}
        x^{(\ell)} \geq \frac{\beta}{n} \, f(M^*) - \frac{1}{\ell} \sum_{k=\ell+1}^{n} x^{(k)},
    \end{equation}
    where we again split the value of subsequent tentative edges into their marginal increases.
    In order to prove the lemma, we show for every $h \in \{t,\dotsc,n\}$ that
    \begin{equation}
        \label{eq:bipartite-matching-inequality}
        \sum_{\ell=t+1}^{n} \frac{t}{\ell-1} \, x^{(\ell)} \geq \sum_{\ell=h+1}^n \biggl( \frac{t}{\ell-1} - \frac{t}{h} \sum_{k=t+1}^{h} \frac{1}{k-1} \biggr) \, x^{(\ell)} + \beta \frac{t}{n} \Biggl( \sum_{\ell=t+1}^{h} \frac{1}{\ell-1} \biggl( 1 - \sum_{k=t+1}^{\ell-1} \frac{1}{k-1} \biggr) \Biggr) \, f(M^*).
    \end{equation}
    We prove the inequality by induction. For $h = t$ the right hand side is equal to the left hand side. Then, assuming that this inequality holds for a fixed $h \in \{t,\dotsc,n-1\}$, we need to show that
    \begin{multline*}
        \sum_{\ell=h+1}^n \biggl( \frac{t}{\ell-1} - \frac{t}{h} \sum_{k=t+1}^{h} \frac{1}{k-1} \biggr) \, x^{(\ell)} + \beta \frac{t}{n} \Biggl( \sum_{\ell=t+1}^{h} \frac{1}{\ell-1} \biggl( 1 - \sum_{k=t+1}^{\ell-1} \frac{1}{k-1} \biggr) \Biggr) \, f(M^*) \\
        \geq \sum_{\ell=h+2}^n \biggl( \frac{t}{\ell-1} - \frac{t}{h+1} \sum_{k=t+1}^{h+1} \frac{1}{k-1} \biggr) \, x^{(\ell)} + \beta \frac{t}{n} \Biggl( \sum_{\ell=t+1}^{h+1} \frac{1}{\ell-1} \biggl( 1 - \sum_{k=t+1}^{\ell-1} \frac{1}{k-1} \biggr) \Biggr) \, f(M^*).
    \end{multline*}
    As we have for every $ t \geq n/2 $ and every $h \in \{t,\dotsc,n-1\}$ that $ \sum_{k=t+1}^{h} \frac{1}{k-1} \leq 1 $, we can apply \Cref{lem:infeasiblegap-submodular-tentative-values} as stated in equation~\eqref{eq:theo-bipartite-matching-restated-lemma} on $x^{(h+1)}$. We obtain that
    {\allowdisplaybreaks
    \begin{align*}
        &\sum_{\ell=h+1}^{n} \biggl( \frac{t}{\ell-1} - \frac{t}{h} \sum_{k=t+1}^{h} \frac{1}{k-1} \biggr) \, x^{(\ell)} + \beta \frac{t}{n} \Biggl( \sum_{\ell=t+1}^{h} \frac{1}{\ell-1} \biggl( 1 - \sum_{k=t+1}^{\ell-1} \frac{1}{k-1} \biggr) \Biggr) \, f(M^*) \\
        &\quad\geq \sum_{\ell=h+2}^{n} \biggl( \frac{t}{\ell-1} - \frac{t}{h} \sum_{k=t+1}^{h} \frac{1}{k-1} \biggr) \, x^{(\ell)} + \beta \frac{t}{n} \Biggl( \sum_{\ell=t+1}^{h} \frac{1}{\ell-1} \biggl( 1 - \sum_{k=t+1}^{\ell-1} \frac{1}{k-1} \biggr) \Biggr) \, f(M^*) \\
        &\qquad + \frac{t}{h} \biggl( 1 - \sum_{k=t+1}^{h} \frac{1}{k-1} \biggr) \biggl( \frac{\beta}{n} \, f(M^*) - \frac{1}{h+1} \, \sum_{\ell=h+2}^{n} x^{(\ell)} \biggr) \\
        &\quad= \sum_{\ell=h+2}^n \biggl( \frac{t}{\ell-1} - \frac{t}{h} \sum_{k=t+1}^{h} \frac{1}{k-1} \biggr) \, x^{(\ell)} - \sum_{\ell=h+2}^{n} \Biggl( \frac{t}{h} \frac{1}{h+1} \biggl( 1 - \sum_{k=t+1}^{h} \frac{1}{k-1} \biggr) \Biggr) x^{(\ell)} \\
        &\qquad + \beta \frac{t}{n} \Biggl( \sum_{\ell=t+1}^{h} \frac{1}{\ell-1} \biggl( 1 - \sum_{k=t+1}^{\ell-1} \frac{1}{k-1} \biggr) \Biggr) \, f(M^*) + \frac{\beta}{n} \frac{t}{h} \biggl( 1 - \sum_{k=t+1}^{h} \frac{1}{k-1} \biggr) \, f(M^*)\\
        &\quad= \sum_{\ell=h+2}^{n} \Biggl( \frac{t}{\ell-1} - \frac{t}{h(h+1)} - \frac{t}{h} \biggl( 1 - \frac{1}{h+1} \biggr) \Biggl( \sum_{k=t+1}^{h} \frac{1}{k-1} \Biggr) \Biggr) x^{(\ell)} \\
        &\qquad + \beta \frac{t}{n} \Biggl( \sum_{\ell=t+1}^{h+1} \frac{1}{\ell-1} \biggl( 1 - \sum_{k=t+1}^{\ell-1} \frac{1}{k-1} \biggr) \Biggr) \, f(M^*) \\
        &\quad= \sum_{\ell=h+2}^{n} \biggl( \frac{t}{\ell-1} - \frac{t}{h+1} \sum_{k=t+1}^{h+1} \frac{1}{k-1} \biggr) \, x^{(\ell)} + \beta \frac{t}{n} \Biggl( \sum_{\ell=t+1}^{h+1} \frac{1}{\ell-1} \biggl( 1 - \sum_{k=t+1}^{\ell-1} \frac{1}{k-1} \biggr) \Biggr) \, f(M^*).
    \end{align*}
    }
    Thus, we have shown~\eqref{eq:bipartite-matching-inequality} and we obtain for $h = n$ that 
    \begin{equation*}
        \sum_{\ell=t+1}^{n} \frac{t}{\ell-1} \, x^{(\ell)} \geq \beta \frac{t}{n} \sum_{\ell=t+1}^{n} \frac{1}{\ell-1} \biggl( 1 - \sum_{k=t+1}^{\ell-1} \frac{1}{k-1} \biggr) \, f(M^*),
    \end{equation*}
    as claimed.
\end{proof}

\begin{lemma}
\label{lem:appendix-optimization}
    For $t = \lfloor e^{\sqrt{2}-2}n \rfloor$, we have that
    \begin{equation*}
        \frac{t}{n} \sum_{\ell=t+1}^{n} \frac{1}{\ell-1} \biggl( 1 - \sum_{k=t+1}^{\ell-1} \frac{1}{k-1} \biggr) \geq e^{\sqrt{2}-2}\bigl(\sqrt{2}-1\bigr) + o(1) \approx 0.23 + o(1).
    \end{equation*}
\end{lemma}

\begin{proof}
    Reindexing the inner summation, we obtain the expression
    \begin{align*}
        f_n(t) &= \frac{t}{n} \sum_{\ell=t+1}^n \frac{1}{\ell-1} \Biggl(1 - \sum_{k=t+1}^{\ell-1} \frac{1}{k-1}\Biggr) = \frac{t}{n} \sum_{\ell=t+1}^n \frac{1}{\ell-1} \Biggl(1 - \sum_{k=t}^{\ell-2} \frac{1}{k}\Biggr).
    \end{align*}
    To obtain a lower bound on $f_n(t)$, we use the standard approximation of the sum by the integral and obtain $\sum_{k=t}^{\ell-2} \frac{1}{k} \leq \int_{t-1}^{\ell-2} \frac{1}{x} \mathrm{d}x = \ln(\ell-2) - \ln(t-1)$. This yields
    \begin{align*}
        f_n(t) &\geq \frac{t}{n} \sum_{\ell=t+1}^n \frac{1}{\ell-1} \biggl(1 - \ln(\ell-2) + \ln(t-1) \biggr) \geq \frac{t}{n} \sum_{\ell=t+1}^n \frac{1}{\ell-1} \biggl(1 - \ln(\ell-1) + \ln(t-1) \biggr)\\
    \end{align*}
    We want to again use the standard approximation of the sum by the integral. We first notice that the term $\frac{1}{\ell-1}\bigl( 1 - \ln(\ell-1) + \ln(t-1)\bigr)$ is decreasing in~$\ell$.
    Hence, we obtain
    \begin{align*}
        \sum_{\ell=t+1}^n \frac{1}{\ell-1} \biggl(1 - \ln(\ell-1) + \ln(t-1) \biggr) &\geq \int_{t+1}^{n+1} \frac{1}{x-1} \biggl(1 - \ln(x-1) + \ln(t-1) \biggr) \mathrm{d}x \\
        &= \int_{t+1}^{n+1} \frac{1}{x-1} \biggl(1 - \ln\bigl(\tfrac{x-1}{t-1}\bigr)\biggr) \mathrm{d}x \\
        &= \int_{t}^{n} \frac{1}{x} \biggl(1 - \ln\bigl(\tfrac{x}{t-1}\bigr)\biggr) \mathrm{d}x.
    \end{align*}
    To solve the integral, we set $\varphi(x) = \ln\bigl(\frac{x}{t-1}\bigr)$ which yields $\varphi'(x) = \frac{1}{x}$. Setting $f(z) = 1-z$, we obtain via integration by substitution
    \begin{align*}
        \int_{t}^{n} \frac{1}{x} \biggl(1 - \ln\bigl(\tfrac{x}{t-1}\bigr)\biggr) \mathrm{d}x &= \int_t^n f(\varphi(x)) \,\varphi'(x)\,\mathrm{d}x \\
        &= \int_{\varphi(t)}^{\varphi(n)} f(z) \,\mathrm{d}z\\
        &= \biggl[u - \frac{1}{2}u^2 \biggr]^{\varphi(n)}_{\varphi(t)}\\
        &= \ln\bigl(\tfrac{n}{t-1}\bigr) - \frac{1}{2}\ln^2\bigl(\tfrac{n}{t-1}\bigr) - \ln\bigl(\tfrac{t}{t-1}\bigr) + \frac{1}{2}\ln^2\bigl(\tfrac{t}{t-1}\bigr).
    \end{align*}
    Putting everything together, we have
    \begin{align}
        f_n(t) &\geq \frac{t}{n}	\biggl(\ln\bigl(\tfrac{n}{t-1}\bigr) - \frac{1}{2}\ln^2\bigl(\tfrac{n}{t-1}\bigr) - \ln\bigl(\tfrac{t}{t-1}\bigr) + \frac{1}{2}\ln^2\bigl(\tfrac{t}{t-1}\bigr) \biggr) \notag\\
        &= \frac{t}{n} \biggl[\ln\bigl(\tfrac{n}{t}\bigr) - \frac{1}{2} \biggl(\ln^2\bigl(\tfrac{n}{t-1}\bigr) -\ln^2\bigl(\tfrac{t}{t-1}\bigr)\biggr)\biggr] \notag \\
        &= \frac{t}{n} \biggl[\ln\bigl(\tfrac{n}{t}\bigr) - \frac{1}{2} \ln^2\bigl(\tfrac{n}{t}\bigr) \biggr] + \frac{t}{2n}\biggl[\ln^2\bigl(\tfrac{n}{t}\bigr) - \ln^2\bigl(\tfrac{n}{t-1}\bigr)-\ln^2\bigl(\tfrac{t}{t-1}\bigr)\biggr] \label{eq:f}
    \end{align}
    Optimizing $t$ for this lower bound is still too complex. We instead consider the following approximation
    \begin{align*}
        g_n(t) =	\frac{t}{n} \biggl[\ln\bigl(\tfrac{n}{t}\bigr) - \frac{1}{2} \ln^2\bigl(\tfrac{n}{t}\bigr)\biggr].
    \end{align*}
    To find the optimal $t$ for this function, we substitute $x = t/n$ and obtain the new function
    \begin{align*}
        h(x) = x\biggl(\ln\bigl(\tfrac{1}{x}\bigr)	- \frac{1}{2}\ln^2\bigl(\tfrac{1}{x}\bigr)\biggr).
    \end{align*}
    The first order condition $h'(x) = 0$ yields the equation
    \begin{align*}
        0 = -1+2\ln\bigl(\tfrac{1}{x}\bigr) -\frac{1}{2} \ln^2\bigl(\tfrac{1}{x}\bigr)	
    \end{align*}
    Substituting $z = \ln\bigl(\frac{1}{x}\bigr)$, we obtain the equation
    \begin{align*}
        0 = -\frac{1}{2}z^2 + 2z -1,
    \end{align*}
    which has the solutions $z_1 = 2 - \sqrt{2}$ and $z_2 = 2 + \sqrt{2}$. As $z = \ln\bigl(\tfrac{1}{x}\bigr) = \ln\bigl(\frac{n}{t}\bigr) \geq 0$, only solution $z_1$ is suitable and we obtain
    \begin{align*}
        \frac{1}{x} &= e^{2-\sqrt{2}}	&&\Leftrightarrow & x = \frac{t}{n} = e^{\sqrt{2}-2} \approx 0.55.
    \end{align*}
    Plugging this into \eqref{eq:f} yields
    \begin{align*}
        e^{\sqrt{2}-2} \biggl[ \bigl(2-\sqrt{2}\bigr) - \frac{1}{2}\bigl(2-\sqrt{2}\bigr)^2 \biggr] + o(1) = e^{\sqrt{2}-2}\bigl(\sqrt{2}-1\bigr) + o(1).
    \end{align*}
    The distinction between $t = e^{\sqrt{2}-2}n$ and $t = \lfloor e^{\sqrt{2}-2}n \rfloor$ vanishes in the limit $n \to \infty$ which implies the result. 
\end{proof}

\newpage

\bibliographystyle{plainnat}
\bibliography{bibliography}

\end{document}